\documentclass[a4paper,USenglish,cleveref,autoref,thm-restate,notab,numberwithinsect,pdfa,nolineno]{socg-lipics-v2021}

\usepackage{amsmath,amsthm,amsfonts}
\usepackage{graphicx}
\usepackage{epstopdf}
\usepackage{thm-restate}
\usepackage{pdfpages}
\usepackage{comment}
\usepackage{nicematrix,tikz}
 \usepackage{transparent}

\newcommand{\Z}{\mathbb{Z}}
\newcommand{\poly}{\text{poly}}
\newcommand{\surf}{\ensuremath{\mathcal S}}
\newcommand{\Set}[1]{\left\{ #1 \right\}}
\newcommand{\Setbar}[2]{\Set{#1 \mid #2}}

\definecolor{blueblack}{rgb}{0,0,.7}
\newcommand{\emphdef}[1]{\textcolor{blueblack}{\emph{#1}}}

\pdfoutput=1
\hideLIPIcs

\title{Computing Shortest Closed Curves\\ on Non-Orientable Surfaces}

\titlerunning{Computing Shortest Closed Curves on Non-Orientable Surfaces}

\author{Denys Bulavka}{Einstein Institute of Mathematics, Hebrew University, Jerusalem 91904, Israel}{denys.bulavka@mail.huji.ac.il}{}{Work done while this author was at Department of Applied Mathematics, Faculty of Mathematics and Physics, Charles University, Prague, Czech Republic. This work was initiated while this author was visiting LIGM, funded by the French National Research Agency grant ANR-17-CE40-0033 (SoS), by the grant no.\ 21-32817S of the Czech Science Foundation (GA\v{C}R) and by Charles University project PRIMUS/21/SCI/014.}

\author{\'Eric Colin de Verdi\`ere}{LIGM, CNRS, Univ.\ Gustave Eiffel, F-77454 Marne-la-Vallée, France}{eric.colin-de-verdiere@univ-eiffel.fr}{}{The work of this author is partly funded by the French National Research Agency grants ANR-17-CE40-0033 (SoS) and ANR-19-CE40-0014 (MIN-MAX).}

\author{Niloufar Fuladi}{LORIA, CNRS, INRIA, Université de Lorraine, F-54000 Nancy, France}{niloufar.fuladi@aol.com}{}{Work done while this author was at LIGM, CNRS, Univ.\ Gustave Eiffel, F-77454 Marne-la-Vallée, France.  The work of this author is partly funded by the French National Research Agency grants ANR-17-CE40-0033 (SoS) and ANR-19-CE40-0014 (MIN-MAX).}

\authorrunning{D.~Bulavka, \'E.~Colin de Verdi\`ere, and N.~Fuladi}


\begin{document}

\maketitle

%%%%%%%%%%%%%%%%%%%%%%%%%%%%%%%%%%%%%%%%%%%%%%%%%%%%%%%%%%%%%
\begin{abstract}
  We initiate the study of computing shortest non-separating simple closed curves with some given topological properties on non-orientable surfaces.  While, for orientable surfaces, any two non-separating simple closed curves are related by a self-homeomorphism of the surface, and computing shortest such curves has been vastly studied, for non-orientable ones the classification of non-separating simple closed curves up to ambient homeomorphism is subtler, depending on whether the curve is one-sided or two-sided, and whether it is orienting or not (whether it cuts the surface into an orientable one).

  We prove that computing a shortest orienting (weakly) simple closed curve on a non-orientable combinatorial surface is NP-hard but fixed-parameter tractable in the genus of the surface.  In contrast, we can compute a shortest non-separating non-orienting (weakly) simple closed curve with given sidedness in $g^{O(1)}\cdot n\log n$ time, where $g$ is the genus and $n$ the size of the surface.

  For these algorithms, we develop tools that can be of independent interest, to compute a variation on canonical systems of loops for non-orientable surfaces based on the computation of an orienting curve, and some covering spaces that are essentially quotients of homology covers.
\end{abstract}

%%%%%%%%%%%%%%%%%%%%%%%%%%%%%%%%%%%%%%%%%%%%%%%%%%%%%%%%%%%%%
\section{Introduction}

In computational topology of graphs on surfaces, much effort has been devoted to computing shortest closed curves with prescribed topological properties on a given surface.  Most notably, the computation of shortest non-contractible, or shortest non-separating, closed curves on a combinatorial surface has been studied, under various scenarios, in at least a dozen papers in the last twenty years~\cite[Table~23.2]{c-ctgs-18}.  Also, algorithms have been given to compute shortest splitting closed curves~\cite{ccelw-scsh-08}, shortest essential closed curves~\cite{ew-csec-10}, shortest closed curves within some non-trivial homotopy class~\cite{cdem-fotc-10}, and shortest closed curves within a given homotopy class~\cite{ce-tnpcs-10}.  In all these cases, the purpose is to compute a shortest closed curve in a given equivalence class, for various notions of equivalence.

Identifying two closed curves on a given surface whenever there is a self-homeo\-mor\-phism of the surface mapping one to the other is certainly one of the most natural equivalence relations.  In particular, this is the most refined relation if we are only given the input surface, and it is relevant in particular in the context of mapping class groups~\cite[Section~1.3.1]{fm-pmcg-11}.  Under this notion, on an orientable surface, any two simple non-separating closed curves are equivalent: Any non-separating simple closed curve cuts the surface into an orientable surface that has (oriented) genus one less than the original surface and with two boundary components.  However, for non-orientable surfaces, it turns out that the classification is subtler: Excluding some low-genus surfaces, a non-separating simple closed curve can be two-sided (have a neighborhood homeomorphic to an annulus) or one-sided (in which case it has a neighborhood homeomorphic to a Möbius band); furthermore it can be orienting (when cutting along it yields an orientable surface with boundary) or not.  

In this paper, we study the complexity of computing shortest non-separating simple closed curves in non-orientable surfaces, under the constraint of being either one-sided or two-sided, either orienting or not, developing, in passing, tools to handle non-orientable surfaces algorithmically.  Before describing our results in detail, we survey previous works.

%%%%%%%%%%%%%%%%%%%%%
\subsection{Previous Works}

One of the most basic and studied questions in topological algorithms for graphs on surfaces is that of computing a shortest non-contractible or non-separating closed curve (the length of such a curve is called \emph{edge-width} in topological graph theory~\cite{ah-irgg-77} or \emph{systole} in Riemannian geometry~\cite{g-frm-83}).  Algorithmically, the simplest setup for graphs on surfaces is that of \emph{combinatorial surfaces}: On a surface~$\surf$, one is given an embedding of a graph~$G$ that is cellular (all faces are homeomorphic to open disks); each edge of the graph has a positive weight.  The goal is to compute a shortest closed walk in~$G$ that is non-trivial on~$\surf$ either in homotopy or in homology.  It turns out that such closed walks are simple, so they are, respectively, a shortest cycle that does not bound a disk on~$\surf$, or that does not separate~$\surf$.

Algorithms for computing shortest non-contractible or non-separating closed curves on surfaces have been developed since the early 1990s~\cite{t-egsnc-90}.  The current fastest algorithm in terms of the size of the input, the number $n$ of vertices, edges, and faces of~$G$, due to Erickson and Har-Peled~\cite{eh-ocsd-04}, runs in $O(n^2\log n)$ time.  However, it is typical to view the genus~$g$ of the surface as a small parameter, and under this perspective it is worth mentioning the algorithms by Cabello, Chambers, and Erickson~\cite{cce-msspe-13}, which runs in $O(g^2n\log n)$ for generic weights, and Fox~\cite{f-sntcd-13}, which runs in $2^{O(g)}\cdot n\log\log n$, but for orientable surfaces only.  We refer to a survey~\cite[Table~23.2]{c-ctgs-18} for many other results.

Other topological properties have also been considered.  Chambers, Colin de Verdière, Erickson, Lazarus, and Whittlesey~\cite{ccelw-scsh-08} study the complexity of computing a shortest ``simple'' closed curve that splits the (orientable) surface into two pieces, neither of which is a disk.  Erickson and Worah~\cite{ew-csec-10} give a near-quadratic time algorithm to compute a shortest essential ``simple'' closed curve on an orientable surface with boundary.  Cabello, DeVos, Erickson, and Mohar~\cite{cdem-fotc-10} provide a near-linear time algorithm (for fixed genus) to compute a shortest ``simple'' closed curve within some (unspecified) non-trivial homotopy class.  In all these problems, one cannot expect in general the output closed curve to be a simple cycle in the input graph~$G$: It sometimes has to repeat vertices and edges of~$G$, but it is \emph{weakly simple}~\cite{cex-dwsp-15} in the sense that it can be made simple by an arbitrary perturbation of the curve on the surface.  In order to store weakly simple curves, both for the output and at intermediate steps of the algorithm, it is convenient to use the dual framework of \emph{cross-metric surface} setting~\cite{ce-tnpcs-10}, in which these curves are really simple.

Very few works devote tools specifically to non-orientable surfaces.  On the combinatorial side, Matou\v{s}ek, Sedgwick, Tancer, and Wagner~\cite{matouvsek2016untangling} carefully describe the result of cutting a non-orientable surface along a simple arc or closed curve, and in particular emphasize that there are various flavors of non-separating simple closed curves.  Very recently, Fuladi, Hubard, and de Mesmay~\cite{fuladi2023short} prove the existence of a canonical system of loops on a non-orientable surface in which each loop has multiplicity $O(1)$, and show that such a system of loops can be computed in polynomial time.  There are some good reasons not to neglect non-orientable surfaces~\cite[Introduction]{fuladi2023short}: Random surfaces are almost surely non-orientable; graphs embeddable on an orientable surface of Euler genus~$2g$ are embeddable on a non-orientable surface of genus $2g+1$, while graphs embeddable on the projective plane can have arbitrarily large orientable genus; non-orientable surfaces appear naturally, e.g., in the graph structure theorem of Robertson and Seymour~\cite{rs-gm16e-03}.

%%%%%%%%%%%%%%%%%%%%%
\subsection{Our Results}

We obtain the following results on orienting (simple) closed curves on non-orientable surfaces:%
\begin{restatable}{theorem}{TNphardness}\label{T:Nphardness}
  It is NP-hard to decide, given a cross-metric surface $(\surf,G^*)$ and an integer~$k$, whether a shortest orienting closed curve on~$(\surf,G^*)$ has length at most~$k$.
\end{restatable}
\begin{theorem}\label{T:Orienting}
  Given a non-orientable cross-metric surface~$(\surf,G^*)$ of genus~$g$ and size~$n$, we can compute a shortest orienting closed curve in~$(\surf,G^*)$ in $2^{O(g)}\cdot n\log n$ time.  Such a shortest closed curve has multiplicity at most two.
\end{theorem}
It turns out that orienting curves are always non-separating, and that their sidedness is prescribed by the genus of the surface (see Lemma~\ref{L:top-charact} below).  In contrast, computing shortest non-orienting (simple) closed curves can be done in polynomial time:%
\begin{restatable}{theorem}{TnonOrienting}\label{T:nonOrienting}
  Given a non-orientable cross-metric surface~$(\surf,G^*)$ of genus~$g$ and size~$n$, we can compute a shortest non-separating non-orienting one-sided (respectively, two-sided) closed curve in~$(\surf,G^*)$ in $O(\operatorname{poly}(g)n\log n)$ time.  More precisely:
  \begin{itemize}
      \item We can compute a shortest non-separating non-orienting one-sided closed curve in $O(g^3n\log n)$ time if $g$ is even, and in $O(g^4n\log n)$ time if $g$ is odd;
      \item we can compute a shortest non-separating non-orienting two-sided closed curve in $O(g^4n\log n)$ time if $g$ is odd, and in $O(g^5n\log n)$ time if $g$ is even.
  \end{itemize}
  Such shortest closed curves have multiplicity at most two.
\end{restatable}
These results implicitly assume that such curves exist (equivalently, $g\ge2$ in the first case and $g\ge3$ in the second one; see Lemma~\ref{L:top-charact}).  They are stated in the cross-metric surface model, but their outputs immediately translate to shortest weakly simple closed curves in the dual combinatorial surface~$(\surf,G)$.

In passing, we develop tools of independent interest for non-orientable surfaces.  First, we give an algorithm to compute an orienting curve in linear time (Proposition~\ref{P:matousek}), refining a construction given by Matou\v{s}ek, Sedgwick, Tancer, and Wagner~\cite{matouvsek2016untangling}.  Second, we introduce an analog of canonical systems of loops for non-orientable surfaces with a new combinatorial pattern, the \emph{standard systems of loops}, and show that we can compute one in asymptotically the same amount of time as canonical systems of loops on orientable surfaces (Proposition~\ref{P:std-sysloops}); we remark that Fuladi, Hubard, and de Mesmay~\cite{fuladi2023short} compute a canonical system of loops in polynomial time, but they do not provide a precise estimate on the running time, which is likely higher than ours.  Such standard systems of loops can be used to compute homeomorphisms between non-orientable surfaces; moreover, some algorithms for surface-embedded graphs rely on canonical systems of loops in the orientable case~\cite{gavoille2023minor,matouvsek2016untangling}, and for non-orientable surfaces they could use standard systems of loops instead. Third, we introduce \emph{subhomology covers}, more compact than the \emph{homology covers} of Chambers, Erickson, Fox, and Nayyeri~\cite{cefn-mcsg-23}%
\footnote{The article by Chambers, Erickson, Fox, and Nayyeri~\cite{cefn-mcsg-23} combines several conference abstracts, including one by Erickson and Nayyeri~\cite{en-mcsnc-11}; the material that we use from~\cite{cefn-mcsg-23} appeared first in~\cite{en-mcsnc-11}.},
but capturing exactly the topological information that we need to classify orienting/non-orienting, one-sided/two-sided closed curves (Section~\ref{S:subhomology-covers}).

%%%%%%%%%%%%%%%%%%%%%
\subsection{Techniques and Organization of the Paper}

After the preliminaries (Section~\ref{S:prelims}), we first show that our two algorithmic results (Theorems~\ref{T:Orienting} and~\ref{T:nonOrienting}) follow from a single common statement (Theorem~\ref{T:common-tool}).  In short, the topological properties that we are considering are homological: Knowing the homology class of a closed curve is enough to decide whether it is separating or not, one-sided or two-sided, orienting or not.  Moreover, the non-separating non-orienting curves with given sidedness are characterized by the fact that they belong to the union of $\poly(g)$ affine subspaces of small codimension in the homology group.  The rest of the paper is devoted to the proof of Theorem~\ref{T:common-tool}.  In Section~\ref{S:matousek}, we give a linear-time algorithm to compute an orienting curve. This serves as a first step to compute a standard system of loops in Section~\ref{S:standard-sysloops}. Section~\ref{S:conversion} provides a way to convert homology from a canonical basis to a standard one.  At this point, given a closed curve, we are able to decide whether it has the desired topological properties efficiently.  In Section~\ref{S:subhomology-covers}, we introduce \emph{subhomology covers}, and then show how to compute a shortest path in the subhomology cover that, when projected, will become the desired closed curve on the surface (Section~\ref{S:subroutine}); this uses a third kind of system of loops, the shortest system of loops~\cite{ew-gohhg-05,c-scgsp-10}.  Section~\ref{S:proof-common-tool} concludes the algorithm.  The NP-hardness proof (Theorem~\ref{T:Nphardness}) is a reduction from the Hamiltonian cycle problem in grid graphs, in the same spirit as the NP-hardness proof of computing a shortest splitting closed curve by Chambers, Colin de Verdière, Erickson, Lazarus, and Whittlesey~\cite{ccelw-scsh-08}; it is presented in Section~\ref{S:hardness}.

%%%%%%%%%%%%%%%%%%%%%%%%%%%%%%%%%%%%%%%%%%%%%%%%%%%%%%%%%%%%%
\section{Preliminaries}\label{S:prelims}

%%%%%%%%%%%%%%%%%%%%%
\subsection{Curves and Graphs on Surfaces}

We use standard terminology of topology of surfaces and graphs drawn on them; see, e.g., Armstrong~\cite{a-bt-83} or, for a more algorithmic perspective, Colin de Verdière~\cite{c-ctgs-18}.  In this paper, the \emphdef{genus} of a surface denotes its Euler genus (which equals twice the standard genus for orientable surfaces).  Unless specified otherwise, surfaces are without boundary.

A simple closed curve on a surface $\surf$ is \emphdef{two-sided} if it has a neighborhood homeomorphic to the annulus. Otherwise, it has a closed neighborhood homeomorphic to the M\"{o}bius band and it is called \emphdef{one-sided}.  A simple closed curve is called \emphdef{non-separating} if the surface (with boundary) we obtain by cutting along it is connected; otherwise the curve is \emphdef{separating}. Note that a separating curve is two-sided. A simple closed curve on a non-orientable surface is called \emphdef{orienting} if by cutting along it, we obtain an orientable surface (with boundary).

A \emphdef{homotopy} between two closed curves is a continuous family of closed curves between them.  A closed curve is \emphdef{contractible} if it is homotopic to a constant closed curve.

Finally, a \emphdef{covering space} of $\surf$ is a topological space $\tilde\surf$ with a continuous map $\pi \colon \tilde \surf \rightarrow \surf$ satisfying the local homeomorphism property, i.e., for every point $x$ in $\surf$ there exists an open neighborhood $U$ of $x$ in $\surf$ and pairwise disjoint open sets $U_1,\dots ,U_d$ in $\tilde \surf$ such that $\bigcup_{i=1}^d U_i = \pi^{-1}(U)$ and $\pi$ restricted to each $U_i$ is a homeomorphism with $U$.  A \emphdef{lift} of a path~$p$ on~$\surf$ is a path~$\tilde p$ on~$\tilde\surf$ such that $\pi\circ\tilde p=p$.

%%%%%%%%%%%%%%%%%%%%%
\subsection{Combinatorial and Cross-Metric Surfaces}

A \emphdef{combinatorial surface} $(\surf,G)$ is the data of a surface~$\surf$ together with a positively edge-weighted graph~$G$ cellularly embedded on~$\surf$; in this model, the curves (or the edges of the graphs) considered later are restricted to be walks in~$G$.  The length of a curve~$c$ is the sum of the weights of the edges of~$G$ used by~$c$, counted with multiplicity.  A closed curve (or a graph) on a combinatorial surface is \emphdef{simple} if it is actually \emph{weakly simple}; namely, if it admits an arbitrarily small perturbation on~$\surf$ that turns it into a simple closed curve (or graph) on~$\surf$.

Weakly simple closed curves and graphs can traverse the same edge of~$G$ or visit the same vertex of~$G$ more than once.  To keep track of these multiplicities, it is often useful to use the concept of \emph{cross-metric surface}~\cite[Section~1.2]{ce-tnpcs-10}, which refines the notion of combinatorial surface in dual form.  A \emphdef{cross-metric surface} $(\surf,G^*)$ is the data of a surface~$\surf$ together with a positively edge-weighted graph~$G^*$ cellularly embedded on~$\surf$; in this model, the graphs and curves considered later are in general position with respect to~$G^*$.  The length of a path~$p$ is the sum of the weights of the edges of~$G^*$ crossed by~$p$, with multiplicity.  The \emphdef{multiplicity} of a path or closed curve on~$(\surf,G^*)$ is the maximum number of times it crosses a given edge of~$G^*$.  Every (weakly) simple closed curve on the combinatorial surface~$(\surf,G)$ corresponds to a simple closed curve of the same length on the cross-metric surface~$(\surf,G^*)$, and conversely.

One can represent cellular graph embeddings on (possibly non-orientable) surfaces using, e.g., graph-encoded maps~\cite{l-gem-82,e-dgteg-03}.  We represent a graph embedded in the cross-metric surface~$(\surf,G^*)$ by storing its \emph{overlay} (or \emph{arrangement}) with~$G^*$, and our algorithms work on this representation.  The \emphdef{size} of a combinatorial or cross-metric surface is the number of vertices, edges, and faces of the underlying graph.

%%%%%%%%%%%%%%%%%%%%%
\subsection{Shortest and Canonical Systems of Loops}

A \emphdef{system of loops} of a surface~$\surf$ is a one-vertex graph~$L$ embedded on~$\surf$ that has a single face, which is homeomorphic to an open disk.  By Euler's formula, $L$ has exactly $g$ loops, where $g$ is the (Euler) genus of~$\surf$.  Cutting $\surf$ along~$L$ results in a $2g$-gon, its \emphdef{polygonal schema}, with edges of its boundary identified in pairs.  The following result follows from earlier works.
\begin{restatable}[\cite{ew-gohhg-05,c-scgsp-10}]{lemma}{Lshsysloops}\label{L:sh-sysloops}
  Let $(\surf,G)$ be a combinatorial surface, orientable or not, of Euler genus~$g$ and size~$n$.  We can compute a set of $2g$~shortest paths on~$G$ such that each non-contractible closed walk in~$G$ intersects at least one of these shortest paths, in $O(gn+n\log n)$ time.
\end{restatable}
\begin{proof}
  We use the following result proved by Erickson and Whittlesey~\cite{ew-gohhg-05} and later extended by Colin de Verdière~\cite{c-scgsp-10}:   Let $(\surf,G^*)$ be a cross-metric surface of Euler genus~$g$ and let $b$ be a point of~$\surf$ not on~$G^*$.  We can compute, in $O(gn+n\log n)$ time, a shortest system of loops based at~$b$ in $(\surf,G^*)$.  Moreover, each loop is the concatenation of two shortest paths with endpoint~$b$, together with a path that crosses a single edge of~$G^*$. 

  Let $(\surf,G^*)$ be the cross-metric surface naturally associated to the combinatorial surface $(\surf,G)$; thus, $G^*$ is the graph dual to~$G$.  Let $b$ be an arbitrary point of~$\surf$ not on the image of~$G^*$.  We apply the above result to $(\surf,G^*)$ and~$b$, obtaining a set of $g$ pairwise disjoint simple loops~$L$ based at~$b$ cutting~$\surf$ into a disk.  Moreover, the set of faces of~$G^*$ visited by the loops is included in the set of faces of~$G^*$ visited by $2g$ shortest paths in~$(\surf,G^*)$, or dually the set of vertices of~$G$ visited by the loops is included in the union of $2g$ shortest paths in~$G$.  Moreover, any non-contractible closed curve in~$G$ must cross at least one loop in~$L$, since every closed curve in a disk is contractible.
\end{proof}

\begin{figure}\def\svgwidth{0.8\linewidth}
    \centering
    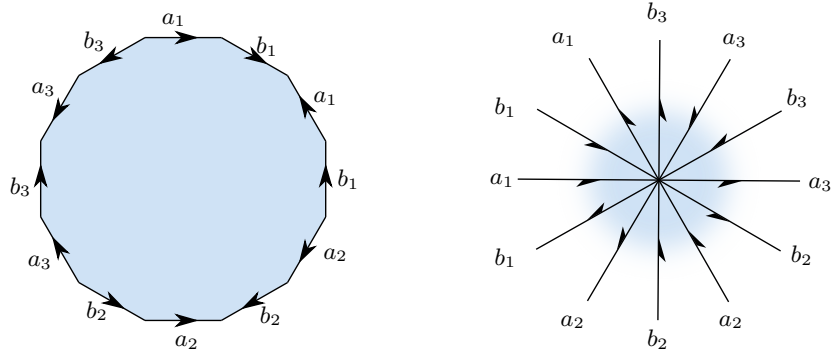
    \caption{Two equivalent views of a canonical system of loops of an orientable surface of Euler genus~6.  Left: A polygonal schema of the form $a_1b_1\bar a_1\bar b_1a_2b_2\bar a_2\bar b_2a_3b_3\bar a_3\bar b_3$.  To reconstruct the surface, one identifies the edges of the polygon in pairs, respecting the orientations of the arrows.  Right: By gluing together the corners of the polygonal schema, we see that all the vertices of the polygon get actually identified into a single vertex on the surface.  The ``half-arrows'' indicate one side of each loop: In the present case, all loops are two-sided, which is reflected by the fact that the right side of a loop when leaving the vertex is still the right side of the loop when it comes back to the vertex.}
    \label{F:canon-orient}
\end{figure}

\begin{figure}\def\svgwidth{0.8\linewidth}
    \centering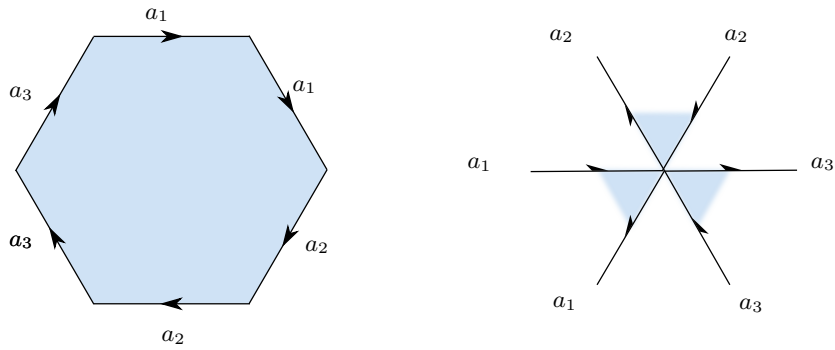
    \caption{Two equivalent views of a canonical system of loops of a non-orientable surface of Euler genus~3.  Left: A polygonal schema of the form $a_1a_1a_2a_2a_3a_3$ with its top (visible) face blue and its bottom (hidden) face white.  Right: The cyclic ordering of the edges around the vertex after identifying the edges of the polygonal schema. All loops are one-sided, which is reflected by the fact that the right side of each loop when leaving the vertex becomes the left side of that loop when it comes back to the vertex.  The colors of the corners indicate the side of the polygon.}
\label{F:canon-nonorient}
\end{figure}

One can describe a polygonal schema (see Figures \ref{F:canon-orient} and~\ref{F:canon-nonorient}) by assigning a symbol and an orientation to each loop in~$L$, and listing the symbols corresponding to the loops encountered along the boundary of the polygonal schema, in clockwise order say, indicating a loop by a bar when encountered with the opposite orientation.  A \emphdef{canonical system of loops for an orientable surface} of Euler genus $2g$ is a system of loops $L$ such that the polygonal schema associated to $L$ has the form $a_1b_1\bar a_1\bar b_1\ldots a_gb_g\bar a_g\bar b_g$ (Figure~\ref{F:canon-orient}).  A \emphdef{canonical system of loops for a non-orientable surface} of genus $g$ is a system of one-sided loops $L$ such that the polygonal schema associated to $L$ has the form $a_1a_1a_2a_2\cdots a_ga_g$ (Figure~\ref{F:canon-nonorient}).  We will use the following result by Lazarus, Pocchiola, Vegter, and Verroust~\cite{lpvv-ccpso-01} to compute a canonical system of loops of an orientable surface:
\begin{lemma}[\cite{lpvv-ccpso-01}]\label{L:lpvv}
  Let $(\surf,G^*)$ be an orientable cross-metric surface with genus $g$ and size~$n$, and let $b$ be an arbitrary point in $(\surf,G^*)$. We can compute a canonical system of loops of $(\surf,G^*)$ based at $b$ in $O(gn)$ time, such that each loop has multiplicity at most four.
\end{lemma}

%%%%%%%%%%%%%%%%%%%%%%%%%%%%%%%%%%%%%
\subsection{Topological Characterizations via Signatures}\label{S:top-charact}

Let $L=(\ell_1,\ldots,\ell_g)$ be any system of loops and let $p$ be a path on~$\surf$ in general position with respect to~$L$.  The \emphdef{signature}~$\sigma(p,L)$ of~$p$ with respect to~$L$ is the vector in~$\Z_2^g$ whose $i$th component is the modulo-$2$ number of crossings of~$p$ with~$\ell_i$.  (Although we will not use this, we remark that $\sigma(p,L)$ expresses the homology of~$p$ in the homology basis dual to~$L$~\cite{cefn-mcsg-23}.)

The following lemma relates the topological type of a curve on a non-orientable surface to its signature with respect to a canonical system of loops.

\begin{restatable}[Schaefer and \v{S}tefankovi\v{c}~\cite{JGAA-580}]{lemma}{Ltopcharact}\label{L:top-charact}
  Let $L$ be a canonical system of loops on a non-orientable surface and $\gamma$ a simple closed curve in general position with respect to~$L$.
  \begin{enumerate}
  \item $\gamma$ is one-sided if and only if $\sigma(\gamma,L)$  has an odd number of $1$ elements;
  \item $\gamma$ is orienting if and only if all the elements in $\sigma(\gamma,L)$ are $1$;
  \item $\gamma$ is separating if and only if all the elements in $\sigma(\gamma,L)$ are $0$.
  \end{enumerate}
\end{restatable}
\begin{proof}
  Given a non-orientable surface $\surf$ of genus $g$, a \emphdef{crosscap decomposition} is a system of simple disjoint one-sided closed curves on $\surf$ such that when $\surf$ is cut along them we obtain the sphere with $g$ boundary components.  This was first introduced by Mohar and is called a \emph{planarizing system of disjoint one-sided curves}, abbreviated as PD1S-system in~\cite{mohar2009genus}. One can turn a canonical decomposition $L$ to a cross-cap decomposition by splitting the loops from the vertex; this can be done close to the vertex such that the crossings between $\gamma$ and the curves in the system do not change after splitting. 

  Let $P$ denote the polygonal schema associated to $L$.
  \begin{enumerate}
    \item By definition of a canonical system of loops (Figure~\ref{F:canon-nonorient}), any loop crossing~$L$ exactly once is one-sided, in the sense that moving along it exactly once reverses the orientation (Figure~\ref{fig:polygon}).  Thus, a simple closed curve~$\gamma$ in general position with respect to~$L$ is one-sided if and only if it crosses~$L$ an odd number of times.

\begin{figure}\def\svgwidth{0.6\linewidth}
    \centering
    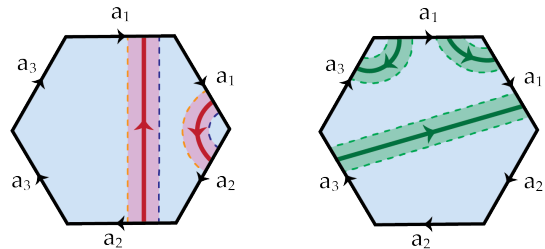
    \caption{Left: A two-sided curve crossing the canonical system of loops an even number of times. Right: A one-sided curve crossing the canonical system of loops an odd number of times.}
    \label{fig:polygon}
  \end{figure}

\item Let $\Theta$ be a canonical decomposition we obtain from $L$ as described above. It is sufficient to show that $\gamma$ crosses each curve in $\Theta$ an odd number of times. By~\cite[Lemma~4]{JGAA-580}, we know that a closed curve is orienting if and only if it crosses each curve in a crosscap decomposition an odd number of times. This finishes the proof.
\item The proof follows the same approach as in the previous case and using~\cite[Lemma~3]{JGAA-580} which states that a closed curve is separating if and only if it crosses each curve in a crosscap decomposition an even number of times. 

As a side note, this statement is true if we replaced $L$ by any system of loops. This is due to the fact that a curve that crosses each loop in such a system an even number of times is null-homologous and a simple closed curve is null-homologous if and only if it is separating.\qedhere
\end{enumerate}
\end{proof}

  Hence, orienting curves are non-separating, and their sidedness is prescribed by the genus.
 
%%%%%%%%%%%%%%%%%%%%%%%%%%%%%%%%%%%%%%%%%%%%%%%%%%%%%%%%%%%%%
\section{Main Technical Result}\label{S:common-tool}

All our algorithms will use the following theorem.  As a motivation, we first show how Theorems~\ref{T:Orienting} and~\ref{T:nonOrienting} follow from it, while postponing its proof until the end of the paper.  
\begin{theorem}\label{T:common-tool}
  Let $(\surf,G^*)$ be a cross-metric surface of Euler genus~$g$ and size~$n$.  Let $k$ be an integer, $\rho:\Z_2^g\to\Z_2^k$ a linear map, and $A\subseteq\Z_2^k$. 

  Then, for some (unspecified) canonical system of loops~$L$ that depends only on~$(\surf,G^*)$, we can compute in $O(g^38^kkn\log n)$ time a shortest closed curve~$c$ on~$(\surf,G^*)$ such that $\rho(\sigma(c,L))\in A$ (if such a curve exists).  This curve is simple and has multiplicity at most two.
\end{theorem}

We emphasize that the canonical system of loops~$L$ is not provided in the input of the algorithm, and is actually never computed explicitly.

\begin{proof}[Proof of Theorem~\ref{T:Orienting}, assuming Theorem~\ref{T:common-tool}]
  Recall that a curve is orienting if and only if it crosses every loop of a canonical system of loops an odd number of times (Lemma~\ref{L:top-charact}).  Thus, to compute a shortest orienting closed curve, we apply Theorem~\ref{T:common-tool} with $k=g$, $\rho$ the identity, and $A$ a single vector, the all-ones vector.
\end{proof}

\begin{proof}[Proof of Theorem~\ref{T:nonOrienting}, assuming Theorem~\ref{T:common-tool}]
  Let $L=\Set{\ell_1,\ldots,\ell_g}$ be the canonical system of loops from Theorem~\ref{T:common-tool} (which we know exists, even if we do not compute it).  Let $c$ be a closed curve.  A loop~$\ell_i$ is \emph{even} (with respect to~$c$) if $\ell_i$ and~$c$ cross an even number of times.  Similarly, $\ell_i$ is \emph{odd} if $\ell_i$ and~$c$ cross an odd number of times.  The \emph{oddity number} of~$c$ is the number of odd loops with respect to~$c$. We first consider the problem of computing shortest (non-separating) non-orienting, one-sided closed curves:
  \begin{itemize}
  \item If $g$ is even, recall from Lemma~\ref{L:top-charact} that a closed curve is non-separating, non-orienting, and one-sided if and only if its oddity number is odd.  To compute a shortest such curve, it suffices to apply Theorem~\ref{T:common-tool} with $k=1$, $\rho(c_1,\ldots,c_g)=\sum_i c_i$, and $A=\Set{1}$.
  
  \item If $g$ is odd, recall from Lemma~\ref{L:top-charact} that a closed curve is non-separating, non-orienting, and one-sided if and only if its oddity number is odd, but different from~$g$.  To compute a shortest such curve, we do the following.  For every $i=1,\ldots,g$, we apply Theorem~\ref{T:common-tool} with $k=2$, $\rho(c_1,\ldots,c_g)=(c_1+\ldots+c_g,c_i)$, and $A=(1,0)$.  This computes a shortest closed curve with odd oddity number such that the $i$th loop of~$L$ is even.  We return the shortest curve over all $i=1,\ldots, g$.
  \end{itemize}
  There remains to prove the theorem for non-separating, non-orienting, two-sided closed curves:
 \begin{itemize}
  \item If $g$ is odd, recall from Lemma~\ref{L:top-charact} that a closed curve is non-separating, non-orienting, and two-sided if and only if the oddity number is even and positive.  To compute a shortest such curve, we do the following.  For every $i=1,\ldots,g$, we apply Theorem~\ref{T:common-tool} with $k=2$, $\rho(c_1,\ldots,c_g)=(c_1+\ldots+c_g,c_i)$, and $A=(0,1)$.  This computes a shortest closed curve with even oddity number such that the $i$th loop of~$L$ is odd.  We return the shortest curve over all $i=1,\ldots, g$.
 
  \item If $g$ is even, recall from Lemma~\ref{L:top-charact} that a closed curve is non-separating, non-orienting, and two-sided if and only if the oddity number is even, positive, and different from~$g$.  To compute a shortest such curve, we do the following.  For every $i\neq j\in\Set{1,\ldots,g}$, we apply Theorem~\ref{T:common-tool} with $k=3$, $\rho(c_1,\ldots,c_g)=(c_1+\ldots+c_g,c_i,c_j)$, and $A=(0,0,1)$.  This computes a shortest closed curve with an even oddity number such that the $i$th loop of~$L$ is even and the $j$th loop of~$L$ is odd.  We return the shortest curve over all $i,j$.\qedhere
\end{itemize}
\end{proof}

%%%%%%%%%%%%%%%%%%%%%%%%%%%%%%%%%%%%%%%%%%%%%%%%%%%%%%%%%%%%%%%%%%%%%%%%%%%%%%
\section{Computing an Orienting Curve}\label{S:matousek}

\begin{restatable}{proposition}{Pmatousek}\label{P:matousek}
    Let $(\surf,G^*)$ be a non-orientable cross-metric surface of size~$n$. We can compute an orienting curve on $(\surf,G^*)$ with multiplicity at most two in time $O(n)$.
\end{restatable}

Matou\v{s}ek, Sedgwick, Tancer, and Wagner~\cite[Proposition 5.5]{matouvsek2016untangling} proved the existence of such curve. We improve on their argument by providing a linear-time algorithm for its computation. We will use the following lemma (below and in Section~\ref{S:proof-common-tool}), which is a variation of a very classical result: Every connected graph with even degrees has an Eulerian cycle.

\begin{restatable}{lemma}{Lsinglecycleoddcrossing}\label{L:single-cycle-odd-crossing}
Let $(\surf,G)$ be a combinatorial surface of size~$n$.  Let $E$ be the set of edges of~$G$.  Let $\mu:E\to\Z_+$ be a map such that (i) for each vertex~$v$ of~$G$, the sum of the values of~$\mu(e)$, for each $e$ incident to~$v$, is even, and (ii) the subgraph of~$G$ induced by the edges~$e$ such that $\mu(e)\ge1$ is connected.  Then one can compute a simple cycle in the dual cross-metric surface~$(\surf,G^*)$ such that, for each edge $e$ of~$G$, the cycle crosses the dual edge~$e^*$ exactly $\mu(e)$ times, in time linear in $n$ plus the sum of the values of~$\mu$.
\end{restatable}
We emphasize that, in Condition~(ii), the considered subgraph cannot have isolated vertices, because it is induced by a set of edges, even though $G$ itself may have vertices incident only to edges~$e$ such that $\mu(e)=0$.
\begin{proof}
     In the following proof we will use interchangeably the $\mu$-value of an edge and its dual. The \emph{$\mu$-degree} of a vertex, as well as of a face, is the sum of the $\mu$-values of its incident edges.
     
    First, without loss of generality we can assume that $\mu(e)\in \{0,1\}$: Indeed, otherwise, we subdivide each dual edge $e^*$ whose $\mu$-value is higher than $1$ into $\mu(e)$ edges, each with $\mu$-value equal to $1$. Moreover, we can assume that each face of $G^*$ has $\mu$-degree $0$, $2$ or $4$: Indeed, otherwise, we iteratively add a diagonal in a face of the dual graph of $\mu$-degree at least six separating a part with $\mu$-degree three from the rest, and set $\mu$ to $1$ on the newly added diagonal; the hypotheses of the lemma still hold.
    
    Within each face in $G^*$ of $\mu$-degree $0$, $2$ or $4$, we draw $0$, $1$ or $2$ disjoint simple paths, respectively, connecting the middle of each boundary edge with $\mu$-value $1$ in an arbitrary way.  All these paths together form simple, pairwise disjoint cycles on the surface, but not necessarily a single cycle. To remedy this, we need to merge the cycles.  At a high level, we perform a search in the \emph{cycle graph}, the graph whose vertices are the cycles, and in which two vertices are adjacent if the corresponding cycles have a face in the dual graph (of the overlay of the cycles and~$G^*$) incident to both of them.  Then, in a second step, we merge the cycles together, by reconnecting each non-root cycle~$c$ to its parent cycle at the location where $c$ was discovered.

    To see that the cycle graph is connected, note that the edges of~$G^*$ cut the cycles into \emph{pieces}.  Let two pieces be \emph{adjacent} if either they are consecutive pieces of the same cycle, or they share a face in the dual graph of the overlay of the cycles and~$G^*$.  By~(ii), this graph is connected.
    
    In more detail, we perform, e.g., a depth-first search rooted at an arbitrary vertex.  To perform the depth-first search, two ingredients are needed: (1) we need to be able to compute the neighbors of a given cycle~$c$ (this is doable in time linear in the number of crossings of~$c$ with~$G^*$), (2) we need to be able to mark cycles as explored or unexplored (for this purpose, with a linear-time preprocessing step, we make each edge that is a piece of a cycle~$c$ point to a separate data structure representing the cycle).  Initially all cycles are unexplored.  We start by exploring the root cycle, and as soon as we discover an unexplored cycle, we mark it as explored and recursively explore it.  This takes linear time.  Whenever a cycle~$c$ is discovered, we take note of the dual face of the overlay of~$G^*$ and of the cycles where this discovery happens.
  
    In the second step, for each of these faces, we reroute $c$ and its parent cycle to merge them into a single cycle, see Figure~\ref{merge_cycles}.  At the end of this process, all cycles are merged into a single one, which crosses every edge of~$G^*$ exactly once, as desired.
\end{proof}
\begin{figure}
  \centering\includegraphics[width=0.8\textwidth]{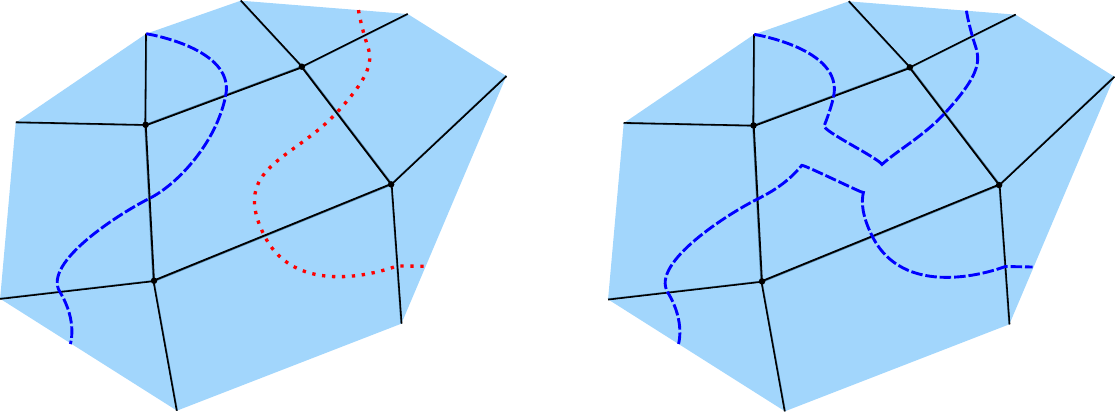}
  \caption{Merging of two disjoint cycles inside a face without increasing the number of intersections with the graph.}
\centering
\label{merge_cycles}
\end{figure}

\begin{proof}[Proof of Proposition~\ref{P:matousek}]
  Let $G$ be the graph dual to~$G^*$.  We first choose an arbitrary orientation of every face of~$G$.  Let $I$ be the set of \emph{inconsistent} edges of~$G$, which bound two faces of~$G$ with inconsistent orientations.  If we start traversing the faces incident to a vertex of $G$ by going around the vertex, in a complete turn, we must change orientations an even number of times to get back to the orientation of the initial face; this implies that each vertex of $G$ is incident to an even number of edges in $I$. Thus the subgraph of~$G$ made of the edges in~$I$ has all its vertices of even degree, and cutting the surface along it yields an orientable surface.  Moreover, we can compute $I$ in linear time.  See Figure~\ref{f:orientation}.  Now, for each edge $e$ of~$G$, let $\mu(e)=1$ if $e\in I$ and $\mu(e)=2$ otherwise.  We apply Lemma~\ref{L:single-cycle-odd-crossing} in the combinatorial surface~$(\surf,G)$ (the hypotheses are obviously satisfied) to compute, in linear time, a simple cycle~$c$ in the cross-metric surface~$(\surf,G^*)$ that crosses each edge~$e$ of~$G^*$ exactly $\mu(e)$ times. 
  
  There remains to prove that $c$ is orienting.  For this purpose, it suffices to exhibit an orientation of each face of the overlay of~$G^*$ and~$c$ in such a way that the inconsistent edges are exactly those arising from~$c$.  We do this as follows.  We first orient the faces of the overlay that touch a vertex of~$G^*$ in the same way as the corresponding face of~$G$ was oriented in the beginning of the proof.  Let $f$ be a face of~$G^*$.  The \emph{subfaces} of~$f$ are the faces of the overlay of~$G^*$ and~$c$ that lie inside it.  Starting from an arbitrary subface of~$f$ that is already oriented, we propagate the orientation to all subfaces of~$f$, in such a way that the subedges of~$c$ are inconsistent (see Figure~\ref{f:orientation}, right).  There is a unique way to do this, because the dual of the subfaces of~$f$ is a tree, and moreover these orientations are compatible with the already selected orientations of the subfaces of~$f$ touching a vertex of~$G$, because each edge~$e^*$ of~$G^*$ is crossed an odd number of times by~$c$ if and only if $e$ lies in~$I$.  By construction, the subedges of~$G^*$ are consistent.
\end{proof}
\begin{figure}
  \centering\includegraphics[width=\textwidth]{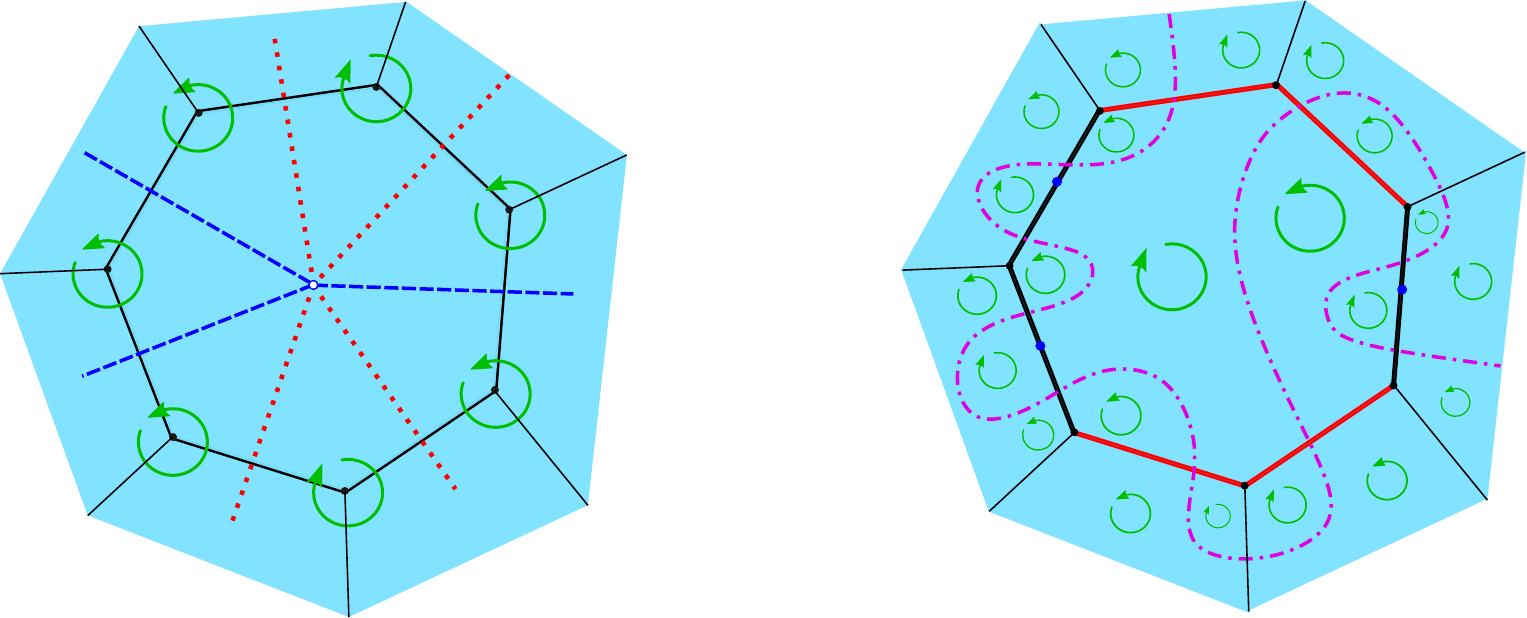}
  \caption{Left: $G$ and $G^*$ with arbitrary orientations of the faces of $G$ and the inconsistent edges of~$G$ (in red, dotted lines); the dual graph~$G^*$ is denoted in black lines. Right: an example of a cycle passing through all the edges and the derived orientation.}
\centering
\label{f:orientation}
\end{figure}
%%%%%%%%%%%%%%%%%%%%%%%%%%%%%%%%%%%%%%%%%%%%%%%%%%%%%%%%%%%%%%%%%%%%%
\section{Computing a Standard System of Loops}\label{S:standard-sysloops}

In this section, we introduce standard systems of loops and show that we can compute one efficiently.  We believe that this result can be of independent interest:  It would perhaps be more natural to compute a canonical system of loops; but on non-orientable surfaces it is only known how to compute one in polynomial time~\cite{fuladi2023short} (the precise worst-case running time is certainly larger than $O(gn)$).  We thus propose this alternate notion of \emph{standard} systems of loops, which can be computed as quickly as the canonical systems of loops for orientable surfaces.  In the next section, we show that for our purposes we can convert from one representation (in terms of parity of crossings) to the other.

A \emphdef{standard system of loops} of a non-orientable surface~$\surf$ of genus~$g$ is a system of loops such that the loops appear in the following order around the boundary of the corresponding polygonal schema (where bar denotes reversal and $p=\lfloor\frac{g-1}2\rfloor$): 
$zza_1b_1\bar a_1\bar b_1\ldots a_pb_p\bar a_p\bar b_p$ if $g$ is odd, and $yw\bar ywa_1b_1\bar a_1\bar b_1\ldots a_pb_p\bar a_p\bar b_p$ if $g$ is even.

\begin{restatable}{proposition}{Pstdsysloops}\label{P:std-sysloops}
  Let $(\surf,G^*)$ be a non-orientable cross-metric surface of genus~$g$ and size~$n$.  In $O(gn)$ time, we can compute a standard system of loops on~$(\surf,G^*)$ such that each loop has multiplicity at most~100.
\end{restatable}
\begin{proof}
  We first compute an orienting simple closed curve~$c$ with multiplicity at most two (Proposition~\ref{P:matousek}).  Let $\surf'$ be the orientable surface with boundary obtained by cutting $\surf$ along~$c$.  We remark that each edge of~$G^*$ corresponds to at most three edges in~$\surf'$, and these edges induce naturally a cross-metric structure on~$\surf'$.

  Let us first assume that $g$ is odd.  Then $c$ is one-sided, and $\surf'$ has a single boundary component~$B$.  Let $u$ be an arbitrary point of~$B$.  Let $\bar\surf'$ be the surface obtained by attaching a disk to the boundary component of~$\surf'$.  In a first step, in $O(gn)$ time, we compute a canonical system of loops~$L$ of~$\bar\surf'$, with basepoint~$v$ in the face incident with~$u$, that has multiplicity at most four and does not cross any edge of~$B$.  We do this as follows.  First, starting from~$\surf'$, we shrink the boundary component~$B$ to a point.  Then, we apply the algorithm by Lazarus, Pocchiola, Vegter, and Verroust~\cite{lpvv-ccpso-01} (Lemma~\ref{L:lpvv}), obtaining a canonical system of loops of the resulting surface based at~$v$.  Finally, we expand back the point to the boundary component~$B$.
  \begin{figure}\def\svgwidth{0.8\linewidth}
    \centering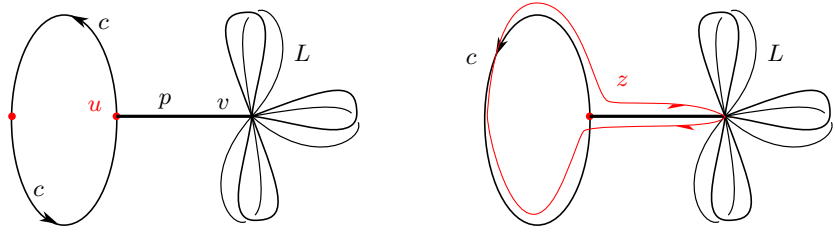
    \caption{The proof of Proposition~\ref{P:std-sysloops}, in the case where $g$ is odd. Left: A pictorial view of the orientable surface~$\surf'$.  The one-sided curve~$c$ is repeated twice on the boundary~$B$ of~$\surf'$.  Right: On~$\surf$, the construction of the loop~$z$, which is one-sided because $c$ is one-sided.}
    \label{F:std1}
  \end{figure}

  Let us now connect $u$ to the basepoint~$v$ with a path~$p$ that arrives at~$v$ at a suitable corner around~$v$ (Figure~\ref{F:std1}).  By this we mean the following:  The cyclic ordering of the edges at~$v$ is $a_1\bar b_1\bar a_1b_1\ldots a_g\bar b_g\bar a_gb_g$ (where bar indicates the origin of an edge, and no bar indicates the target of an edge), see Figure~\ref{F:canon-orient}; we make $p$ arrive between two consecutive groups of the form $a_i\bar b_i\bar a_ib_i$.  For this purpose, let $f$ be the face containing the basepoint~$v$ of~$L$; it is incident with~$u$.  The loops cut~$f$ into \emph{subfaces}; let $f'$ be the subface incident with~$u$.  Starting from~$u$ in~$f'$, we draw a path~$p$ that goes to a portion of~$L$ that lies on the boundary of~$f'$, and then runs along~$L$ (possibly exiting~$f$) until we get to the basepoint~$v$ at a suitable corner.  Because, when running along the boundary of the polygonal schema, the suitable corners appear regularly, every four corners (see Figure~\ref{F:canon-orient}), this can be done by running along at most two loops; each of these loops has multiplicity at most four on~$\surf'$, and thus $p$ has multiplicity at most eight on~$\surf'$.  This path~$p$ is computed in $O(n)$ time.

  Finally, the set of loops~$L$, together with the loop~$z$ that is the concatenation of the reversal of~$p$, $c$, and (a slightly translated copy of)~$p$, is a standard system of loops of~$\surf$, which is computed in $O(gn)$ time.  Moreover, every edge of~$G^*$ is crossed at most twice by~$c$ and at most 24 times by~$p$, so the multiplicity of each loop is at most~50 with respect to~$G^*$.

  \begin{figure}\def\svgwidth{0.8\linewidth}
    \centering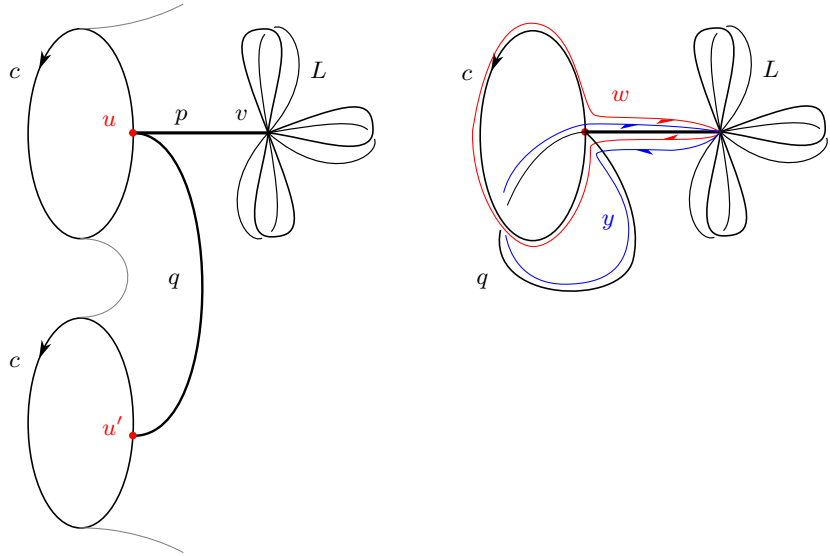
    \caption{The proof of Proposition~\ref{P:std-sysloops}, in the case where $g$ is even. Left: A pictorial view of the orientable surface~$\surf'$.  The two-sided curve~$c$ corresponds to two boundary components of~$\surf'$.  Right: On~$\surf$, the construction of the loops~$w$ and~$y$. }
    \label{F:std2}
  \end{figure}

  There remains to consider the case where $g$ is even; see Figure~\ref{F:std2}.  In that case, $\surf'$ has exactly two boundary components.  We choose an arbitrary basepoint~$u$ on one of the boundary components.  Let $u'$ be the unique point on the other boundary component such that $u$ and~$u'$ are identified on~$\surf$.  We first compute a path~$q$ of multiplicity one on~$\surf'$ between $u$ and~$u'$.  Let $\surf''$ be the surface obtained from~$\surf'$ by cutting along~$q$; it has a single boundary component.  Let $u''$ be a point of~$\surf''$ mapping to~$u$ or~$u'$ after regluing to obtain~$\surf'$.  Then we proceed as in the case where $g$ is odd, using $u''$ in place of~$u$.  The final standard system of loops is made (after a small perturbation) of the canonical system of loops of~$\surf''$, the loop~$y$ that is the concatenation of the reversal of~$p$, $q$, and $p$ (which is one-sided because the closed curve corresponding to~$q$ is one-sided; indeed, because $c$ is two-sided and cuts $\surf$ into an orientable surface, any closed curve crossing~$c$ once must be one-sided, since otherwise $\surf$ would be orientable), and the loop~$w$ that is the concatenation of the reversal of~$p$, $c$, and~$p$ (which is two-sided because $c$ is two-sided).  The multiplicity of each loop is at most~100 (because now we have to take into account the path~$q$, which cuts some edges of~$\surf'$ into two subedges).
\end{proof}

%%%%%%%%%%%%%%%%%%%%%%%%%%%%%%%%%%%%%%%%%%%%%%%%%%%%%%%%%%%%%
\section{Converting Between Canonical and Standard Signatures}\label{S:conversion}

Our next lemma describes a matrix allowing to change the coordinates in homology, from the standard basis to some canonical one.  The motivation is that the topological properties we need are best expressed in terms of crossings with a canonical system of loops, while we can more easily compute a standard system of loops.

\begin{restatable}{lemma}{Lchangebasis}\label{L:change-basis}
  Assume $L'$ is a \emph{standard} system of loops on a non-orientable surface.  There exists a \emph{canonical} system of loops~$L$ such that, for each path~$p$,  we have $\sigma(p,L)=\varphi(\sigma(p,L'))$, where $\varphi:\Z_2^g\to\Z_2^g$ is the invertible linear map described by the following matrices depending on the parity of $g$.

{\footnotesize
\[
A_{\text{odd}} = \begin{bmatrix}
    1 & 1 & 0 & 0 & 0 & 0 & \cdots & 0 & 0 & 0 \\
    1 & 1 & 1 & 0 & 0 & 0 & \cdots & 0 & 0 & 0 \\
    1 & 0 & 1 & 1 & 0 & 0 & \cdots & 0 & 0 & 0 \\
    1 & 0 & 1 & 1 & 1 & 0 & \cdots & 0 & 0 & 0 \\
    1 & 0 & 1 & 0 & 1 & 1 & \cdots & 0 & 0 & 0 \\
    1 & 0 & 1 & 0 & 1 & 1 & \cdots & 0 & 0 & 0 \\
    1 & 0 & 1 & 0 & 1 & 0 & \cdots & 0 & 0 & 0 \\
    \vdots & \vdots & \vdots & \vdots & \vdots & \vdots & \ddots & \vdots & \vdots & \vdots \\
    1 & 0 & 1 & 0 & 1 & 0 & \cdots & 1 & 1 & 1 \\
    1 & 0 & 1 & 0 & 1 & 0 & \cdots & 1 & 0 & 1 
\end{bmatrix}
\qquad
A_{\text{even}} = \begin{bmatrix}
    1 & 1 & 0 & 0 & 0 & 0  & \cdots & 0 & 0 & 0\\
    0 & 1 & 1 & 0 & 0 & 0 & \cdots &0 & 0 & 0 \\
    0 & 1 & 1 & 1 & 0 & 0  & \cdots & 0 & 0 & 0\\
    0 & 1 & 0 & 1 & 1 & 0  & \cdots & 0 & 0 & 0\\
    0 & 1 & 0 & 1 & 1 & 1  & \cdots & 0 & 0 &0\\
    0 & 1 & 0 & 1 & 0 & 1  & \cdots & 0 & 0 &0\\
    0 & 1 & 0 & 1 & 0 & 1  & \cdots & 0 & 0 &0\\
    \vdots & \vdots & \vdots & \vdots & \vdots & \vdots& \ddots & \vdots  & \vdots & \vdots \\
    0 & 1 & 0 & 1 & 0 & 1  & \cdots & 1&1 & 1 \\
    0 & 1 & 0 & 1 & 0 & 1 & \cdots & 1 & 0 &1\\
\end{bmatrix}
\]
}
\end{restatable}

We remark that we actually never compute the canonical system of loops~$L$.

\begin{proof}
The proof uses the cut-and-paste technique used in the proof of the classification of surfaces. We start with the polygonal schema associated to $L'$ and at each step we cut along a diagonal and paste along one of the sides of the polygon. Each diagonal can be seen as a concatenation of the edges in the polygonal schema and becomes a side in the new polygonal schema that we obtain after pasting.

This will prove that $\varphi^{-1}$ is given by the following matrices, each being the inverse of the one stated in the lemma, as a simple computation shows.

{\footnotesize
\[
A_{\text{odd}}^{-1} = \begin{bmatrix} 
    1 & 1 & 1 & 1 & 1 & 1 & \cdots & 1 & 1 & 1 \\
    0 & 1 & 1 & 1 & 1 & 1 & \cdots & 1 & 1 & 1 \\
    1 & 1 & 0 & 0 & 0 & 0 & \cdots & 0 & 0 & 0 \\
    0 & 0 & 0 & 1 & 1 & 1 & \cdots & 1 & 1 & 1 \\
    0 & 0 & 1 & 1 & 0 & 0 & \cdots & 0 & 0 & 0 \\
    0 & 0 & 0 & 0 & 0 & 1 & \cdots & 1 & 1 & 1 \\
    \vdots & \vdots & \vdots & \vdots & \vdots & \vdots & \ddots & \vdots & \vdots & \vdots \\
    0 & 0 & 0 & 0 & 0 & 0 & \cdots & 0 & 0 & 0 \\
    0 & 0 & 0 & 0 & 0 & 0 & \cdots & 0 & 1 & 1 \\
    0 & 0 & 0 & 0 & 0 & 0 & \cdots & 1 & 1 & 0 \\
\end{bmatrix}
\qquad
A_{\text{even}}^{-1} = \begin{bmatrix} 
    1 & 1 & 1 & 1 & 1 & 1 & \cdots & 1 & 1 & 1 \\
    0 & 1 & 1 & 1 & 1 & 1 & \cdots & 1 & 1 & 1 \\
    0 & 0 & 1 & 1 & 1 & 1 & \cdots & 1 & 1 & 1 \\
    0 & 1 & 1 & 0 & 0 & 0 & \cdots & 0 & 0 & 0 \\
    0 & 0 & 0 & 0 & 1 & 1 & \cdots & 1 & 1 & 1 \\
    0 & 0 & 0 & 1 & 1 & 0 & \cdots & 0 & 0 & 0 \\
    \vdots & \vdots & \vdots & \vdots & \vdots & \vdots & \ddots & \vdots & \vdots & \vdots \\
    0 & 0 & 0 & 0 & 0 & 0 & \cdots & 0 & 0 & 0 \\
    0 & 0 & 0 & 0 & 0 & 0 & \cdots & 0 & 1 & 1 \\
    0 & 0 & 0 & 0 & 0 & 0 & \cdots & 1 & 1 & 0 \\
\end{bmatrix}
\]

}

\begin{figure}
  \def\svgwidth{\linewidth}\scriptsize
    \centering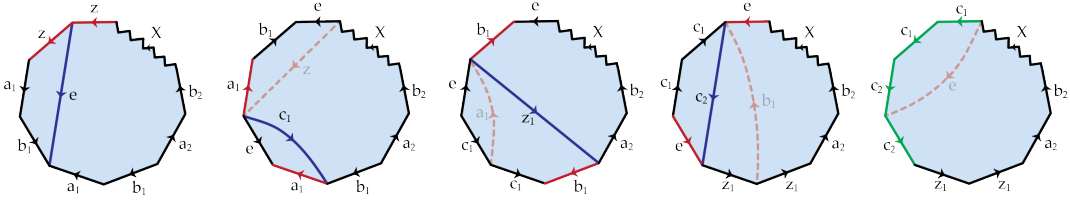
  \caption{The cut-and-paste process to go from $L'$ to $L_1$ in the case where $g$ is odd. At each polygon, the blue diagonal introduces the cut and the red edge denotes the edge along which we paste. Here, $X$ denotes the loops $\bar a_2\bar b_2a_3b_3\bar a_3 \bar b_3\ldots a_pb_p\bar a_p \bar b_p$.}
\centering
\label{cut and pasting}
\end{figure}

First let us consider the case where $g=2p+1$ is odd. We first show that using $4$ cut-and-paste moves we can transform the polygonal schema $L': zza_1b_1\bar a_1\bar b_1\ldots a_pb_p\bar a_p\bar b_p$ to $L_1: c_1c_1c_2c_2z_1z_1a_2b_2\bar a_2\bar b_2\ldots a_p b_p \bar a_p \bar b_p$, see Figure~\ref{cut and pasting}. We can see that $c_1=e+a_1=z+b_1$ in which $u+v$ denotes the concatenation of two loops $u$ and $v$, ignoring ordering and orientation (which is irrelevant as far as crossing numbers are concerned); in other words, for any path~$p$, we have $\sigma(p,c_1)=\sigma(p,z)+\sigma(p,b_1)$.  Similarly, $c_2=e=z+a_1+b_1$ and $z_1=e+b_1=z+a_1$.  We can see that these cut-and-paste scenarios can be described by the following matrix: let $M_1$ be a $g\times g$ matrix with columns $m_i$ for $1\leq i \leq g$ such that $m_1=(1,0,1,0,\ldots,0)$ (corresponding to~$c_1$), $m_2=(1,1,1,0,\ldots, 0)$ (corresponding to~$c_2$), and $m_3=(1,1,0,0,\ldots,0)$ (corresponding to~$z_1$).  For $4\leq i\leq g$, let $m_i$ be the $g\times 1$ matrix with every element $0$ except its $i$th element that is $1$. Note that the entries of this matrix come from $\Z_2$, the rows correspond to the elements in $L'$ and the columns correspond to the elements in $L_1$.

\begin{figure}
  \def\svgwidth{\linewidth}\scriptsize
    \centering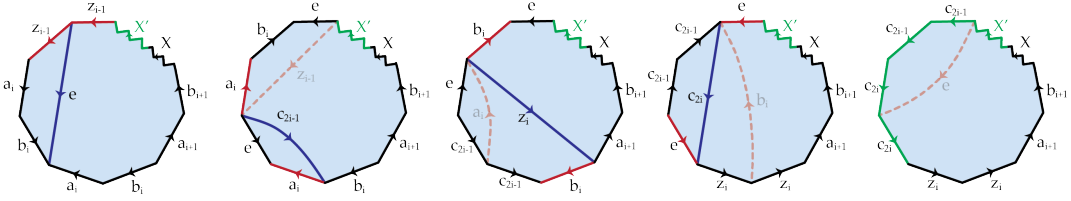
  \caption{The cut-and-paste process to go from $L_{i-1}$ to $L_i$ in the case where $g$ is odd. Here, $X:\bar a_{i+1}\bar b_{i+1}\ldots a_pb_p\bar a_p \bar b_p$ and $X': c_1c_1\ldots c_{2i-2}c_{2i-2}$.}
\centering
\label{cut and pasting2}
\end{figure}
Let $L_{i-1}$ denote the polygonal schema $c_1c_1\ldots c_{2i-2}c_{2i-2}z_{i-1}z_{i-1}a_ib_i\bar a_i \bar b_i\ldots a_pb_p$. Mimicking the same cut-and-paste moves depicted in Figure~\ref{cut and pasting} for the sides $z_{i-1},a_i$ and $b_i$ for $i>1$ transforms $L_{i-1}$ to the polygonal schema $L_i: c_1c_1\ldots c_{2i}c_{2i}z_iz_ia_{i+1}b_{i+1}\bar a_{i+1}\bar b_{i+1}\ldots a_p b_p \bar a_p \bar b_p$, see Figure~\ref{cut and pasting2}.  We denote the corresponding matrix by $M_{i}$; this matrix changes the basis from $L_{i}$ to $L_{i-1}$, keeping all the elements unchanged except $z_{i-1}$, $a_i$, and $b_i$: $c_{2i-1}=z_{i-1}+b_i$, $c_{2i}=z_{i-1}+a_i+b_i$, and $z_i=z_{i-1}+a_i$.  Continuing this process, we obtain the canonical polygonal schema $L$ with the final matrix being given by $A^{-1}_{\text{odd}}=M_1M_2\cdots M_p$.

\begin{figure}
   \def\svgwidth{.5\linewidth}\scriptsize
    \centering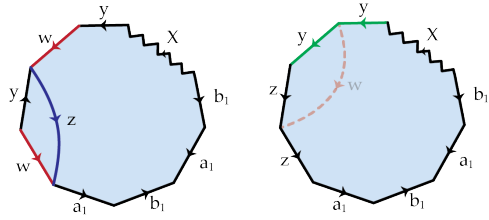
  \caption{The additional cut-and-paste move in the case where $g$ is even.}
\centering
\label{cut and pasting 3}
\end{figure}

In the case where $g=2p$ is even, we need one additional type of cut to transform $L': yw\bar ywa_1b_1\bar a_1\bar b_1\ldots a_{p-1}b_{p-1}\bar a_{p-1}\bar b_{p-1}$ to $L_1: yyzza_1b_1\bar a_1\bar b_1\ldots a_{p-1}b_{p-1}\bar a_{p-1}\bar b_{p-1}$, see Figure~\ref{cut and pasting 3}. Then we can use the four types of cuts introduced for the case where the genus is odd to turn $L_1$ to $L$.

One can check that in this case the corresponding matrix is the matrix
$A^{-1}_{\text{even}}$.
\end{proof}

%%%%%%%%%%%%%%%%%%%%%%%%%%%%%%%%%%%%%%%%%%%%%%%%%%%%%%%%%%%%%%%%%%%%%
\section{Subhomology Covers}\label{S:subhomology-covers}

In this section, we introduce subhomology covers, which are covering spaces related to quotients of the homology group.  They generalize \emph{cyclic double covers} introduced by Erickson~\cite{e-sntcd-11} and used by Borradaile, Chambers, Fox, and Nayyeri~\cite{bcfn-mchbs-17} and are inspired from \emph{homology covers} by Chambers, Erickson, Fox, and Nayyeri~\cite{cefn-mcsg-23}, but there are important differences with the latter: They can capture not necessarily the homology group, but arbitrary quotients of the homology group (which allows for faster algorithms), and they are defined on surfaces without boundary. This tool is not restricted to non-orientable surfaces and we present it for arbitrary surfaces.  Our construction is inspired from the \emph{voltage construction} by Gross and Tucker~\cite[Chapter~4]{gross-tucker87-graph}.

Throughout this section, let $(\surf,G)$ be a combinatorial surface (orientable or not) of Euler genus~$g$ and size~$n$, $L$ a system of loops in general position with respect to $G$ (we will only need the case where $L$ is the standard system of loops, computed in Proposition~\ref{P:std-sysloops}, but we do not assume this for the construction).  Furthermore, let $k$ be an integer and $\rho:\Z_2^g\to\Z_2^k$ a linear map.  We define $\alpha:E\to\Z_2^k$ by $\alpha(e)=\rho(\sigma(e,L))$.
\begin{lemma}
  The map~$\alpha$ satisfies the \emph{Kirchhoff voltage law} on every face: For every face of $G$ with boundary edges $e_1,\dots,e_m$, we have that $\sum_{i=1}^m\alpha(e_i)=0$. 
\end{lemma}
\begin{proof}
  The boundary of a face bounds a disk and consequently is separating.  Thus every loop in~$L$ intersects it an even number of times.
\end{proof}

We define the graph $\tilde G$ as a graph with vertices $(v,\nu)$, one for every vertex $v$ in $G$ and $\nu$ in $\Z^k_2$. For each edge $e$ in $G$, connecting the vertices $u$ and $v$, and every $\nu$ in $\Z_2^k$, there is an edge in $\tilde G$ connecting $(u,\nu)$ and $(v,\nu + \alpha(e))$. We now observe that each facial cycle in~$G$ (which bounds a face of~$G$ in~$\surf$) lifts to a cycle in~$\tilde G$, by the Kirchhoff voltage law.  By attaching a disk to the lift of each such cycle, we obtain a combinatorial surface $(\tilde\surf,\tilde G)$, which is naturally a (possibly non-connected) covering space of~$\surf$ and is called the \emphdef{subhomology cover associated to $\alpha$}.  An alternate way to build $(\tilde\surf,\tilde G)$ is to cut~$\surf$ along~$L$, obtaining a disk~$D$, and to glue together $2^k$ copies of these disks, in such a way that the copy~$\nu$ gets attached via a lift of loop~$\ell_i\in L$ to the copy~$\nu+\rho(s_i)$, where $s_i\in\Z_2^g$ has a single non-zero entry, the $i$th one (Figure~\ref{fig:subhomology cover}); however, we will not need this equivalence.

\begin{figure}
  \def\svgwidth{.9\linewidth}\scriptsize
    \centering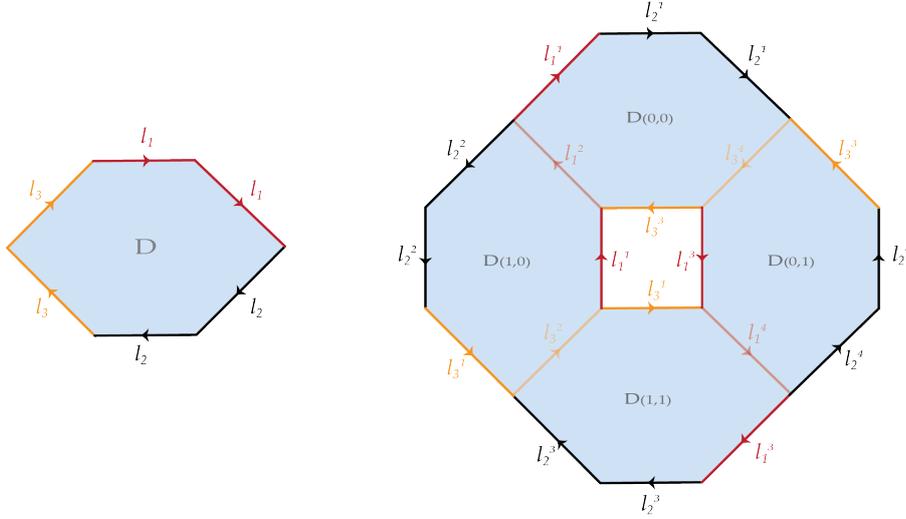
  \caption{Left: A non-orientable surface of genus $3$ cut along a canonical system of loops. Right: A non-orientable 4-sheeted covering space of the surface at left with genus $6$ given by $\rho:\Z_3^k\rightarrow \Z_2^k$; $\rho(s_1)=(1,0)$, $\rho(s_2)=(0,0)$ and $\rho(s_3)=(0,1)$.  The notation $l_i^j$ indicates the $j$th lift of the loop~$l_i$.}
\centering
\label{fig:subhomology cover}
\end{figure}
\begin{lemma}\label{L:rhocovers-algo}
  Assume that one is given the combinatorial surface~$(\surf,G)$ together with the map~$\alpha$.  Then, in $O(2^kkn)$ time, we can compute the combinatorial surface $(\tilde\surf,\tilde G)$ that is the subhomology cover of~$(\surf,G)$ associated to $\alpha$.  Moreover, $\chi(\tilde \surf) = 2^k\chi(\surf)$.
\end{lemma}
Again, the system of loops~$L$ is not part of the input; only the map $\alpha$ is specified.
\begin{proof}
  The vertices of $\tilde G$ are given by the lifts of the vertices of $G$, that is, by pairs $(v,\nu)$ for every vertex~$v$ of~$G$ and every $\nu \in \Z_2^k$. Every edge $e$ from $v$ to $w$ in $G$ lifts, given $\nu\in \Z_2^k$, to an edge from $(v,\nu)$ to $(w,\nu + \alpha(e))$. Finally, each facial cycle in~$G$, incident to vertex~$v$ of~$G$, lifts to $2^k$ faces of~$\tilde G$ in~$\tilde\surf$, each incident to $(v,\nu)$ for $\nu\in\Z_2^k$, as explained above.
 
  The combinatorial map of $\tilde G$ can thus be computed in $O(2^kkn)$ time.  A more explicit construction would depend on the data structure used.  For example, in the graph-encoded map data structure~\cite{l-gem-82,e-dgteg-03}, each flag~$f$ of~$G$ incident to vertex~$v$ corresponds to $2^k$ flags of~$\tilde G$, denoted $\Setbar{(f,\nu)}{\nu\in\Z_2^k}$ where $(f,\nu)$ has associated vertex~$(v,\nu)$, and we can easily connect the flags via their three involutions, in overall $O(2^kkn)$ time.

  The last claim follows from the fact that every vertex, edge, and face of~$G$ on~$\surf$ corresponds to $2^k$ vertices, edges, and faces of~$\tilde G$ on~$\tilde\surf$, respectively.
\end{proof}

\begin{lemma}\label{l:endpoint}
  Let $c$ be a closed walk in~$G$, and $\tilde c$ be a lift of~$c$ in~$\tilde G$, with endpoints $(v,\nu)$ and~$(v,\nu')$.  Then $\rho(\sigma(c,L)) = \nu' - \nu$.
\end{lemma}
\begin{proof}
  By construction, if $e$ connects $w$ with~$w'$, then each lift of~$e$ connects vertices $(w,\eta)$ to $(w',\eta+\alpha(e))$ for some $\eta\in\Z^k_2$.  Thus $\nu'-\nu$ equals the sum, over all edges~$e$ of~$c$, of $\alpha(e)=\rho(\sigma(e,L))$.  This, in turn, equals $\rho(\sigma(c,L))$.  
\end{proof}
Although we do not make an explicit use of it here, it is worth mentioning that our techniques are (co)homological in nature. Indeed, we could have defined the signature map $\rho$ as a map from the first (co)homology group, over $\Z_2$, instead of being from $\Z_2^g$. The difference between these two being that the later assumes a choice of basis, i.e., a system of loops. Different choices of basis will indeed produce isomorphic covering spaces since this one is uniquely determined by the normal subgroup of the fundamental group induced by the kernel of the map $\rho$.

%%%%%%%%%%%%%%%%%%%%%%%%%%%%%%%%%%%%%%%%%%%%%%%%%%%%%%%%%%%%%
\section{Computing Shortest Closed Walks with Restriction on the Signature}\label{S:subroutine}

Our algorithm will use the subroutine given in the following proposition, which is a variation on the strategy used earlier by Chambers, Erickson, Fox, and Nayyeri~\cite[Section~5.2]{cefn-mcsg-23}.  It is a first step towards the proof of Theorem~\ref{T:common-tool}; however, instead of computing a simple closed curve in the cross-metric surface, we only compute a closed walk in the dual combinatorial surface.

\begin{proposition}\label{P:subroutine}
  Let $(\surf,G)$ be a combinatorial surface, orientable or not, of Euler genus~$g$ and size~$n$.  Moreover, let   $k$ be an integer, $\rho:\Z_2^g\to\Z_2^k$ a linear map, and $A\subseteq\Z_2^k$.  Assume that for each edge~$e$ of~$G$, the value of $\rho(\sigma(e,L))$ is given, for a fixed system of loops~$L$ of~$\surf$ in general position with respect to~$G$.  Given this, we can compute in $O(g^38^kkn\log n)$ time a shortest closed curve~$c$ in~$G$ (a shortest closed walk) such that $\rho(\sigma(c,L))\in A$.
\end{proposition}

For the proof, we start by computing the subhomology cover~$(\tilde\surf,\tilde G)$ associated to $\alpha(e)=\rho(\sigma(e,L))$ using Lemma~\ref{L:rhocovers-algo}.  We also apply Lemma~\ref{L:sh-sysloops} to compute a set of $2g$ shortest paths $p_1,\ldots,p_{2g}$ in~$G$ that intersects every non-contractible closed walk in~$G$.

The following lemma is inspired by Chambers, Erickson, Fox, and Nayyeri~\cite[Lemma~5.4]{cefn-mcsg-23}.%
\begin{restatable}{lemma}{Lsubroutine}\label{L:subroutine}
  Some shortest closed curve~$c$ in~$G$ such that $\rho(\sigma(c,L))\in A$ has the following form: It is the projection of a shortest path in~$(\tilde \surf,\tilde G)$ that starts with a subpath of a lift of some~$p_i$ and is otherwise disjoint from that lift.
\end{restatable}
\begin{figure}\def\svgwidth{0.8\linewidth}
    \centering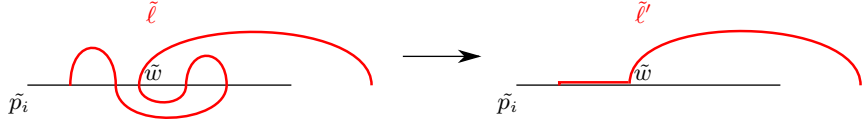
    \caption{Proof of Lemma~\ref{L:subroutine}: the construction of~$c'$.}
    \label{F:shortening}
\end{figure}
\begin{proof}
  If $A$ contains the zero vector, then any trivial closed curve (reduced to a single vertex) satisfies the desired property, so the lemma is trivially true.  So, henceforth, we assume that $A$ does not contain the zero vector.  
  
  Let $c$ be a shortest closed curve in~$G$ such that $\rho(\sigma(c,L))\in A$.  Because $A$ does not contain the zero vector, $c$ is non-contractible and must thus meet some path~$p_i$.  We turn $c$ into a loop~$\ell$ by basing it at an arbitrary vertex of~$c\cap p_i$.  Let $\tilde p_i$ be a path in~$\tilde G$ that is a lift of~$p_i$ in the subhomology cover~$\tilde \surf$, and let $\tilde\ell$ be a lift of~$\ell$ that starts on~$\tilde p_i$. See Figure~\ref{F:shortening}.  Let $\tilde w$ be the last vertex of~$\tilde\ell\cap\tilde p_i$ encountered when traversing~$\tilde\ell$.  Let $\tilde\ell'$ be the path in~$\tilde G$ with the same endpoints as~$\tilde\ell$ obtained from~$\tilde\ell$ by replacing its part before~$\tilde w$ with the subpath of~$\tilde p_i$ with the same endpoints.  Finally, let $\ell'$ be the projection of~$\tilde\ell'$ on~$\surf$, and $c'$ be the associated closed curve.

  Since $\tilde\ell$ and~$\tilde\ell'$ have the same endpoints, Lemma~\ref{l:endpoint} implies that $\rho(\sigma(c',L))=\rho(\sigma(c,L))$, the latter being in~$A$ by hypothesis.  Moreover, $\tilde\ell'$ is no longer than~$\tilde\ell$, because $\tilde p_i$ is a shortest path, since any lift of a shortest path is itself a shortest path.  Thus, $c'$ is no longer than~$c$, and satisfies the desired properties.
\end{proof}

We will also need the following immediate consequence of results by Erickson, Fox, and Lkhamsuren~\cite{efl-hmcpf-18} and Cabello, Chambers, and Erickson~\cite{cce-msspe-13}.
\begin{restatable}{lemma}{Lmultipleshortestpaths}\label{L:multiple-shortest-paths}
  Given a combinatorial surface~$(\surf,G)$, orientable or not, of Euler genus~$g$ and size~$n$, with a distinguished face~$b$, one can, after an $O(g^2n\log n)$-time preprocessing, compute the distance from any given vertex incident with~$b$ to any other vertex in $O(\log n)$ time.  
\end{restatable}
\begin{proof}
  The case where the surface is orientable is precisely the result by Erickson, Fox, and Lkhamsuren~\cite[Theorem~5.1]{efl-hmcpf-18}.  Moreover, as explained by Cabello, Chambers, and Erickson~\cite[Section~4.4]{cce-msspe-13}, the non-orientable case reduces to the orientable one, by passing to the orientable double cover. 
\end{proof}

\begin{proof}[Proof of Proposition~\ref{P:subroutine}]
  As indicated above, we first compute the subhomology cover~$(\tilde \surf,\tilde G)$ in $O(2^kkn)$ time (Lemma~\ref{L:rhocovers-algo}), and the paths~$p_1,\ldots,p_{2g}$ in $O(gn+n\log n)$ time (Lemma~\ref{L:sh-sysloops}).
  
  Fix $i=1,\ldots,2g$.  We show below how to compute a shortest path in~$(\tilde \surf,\tilde G)$ that starts with a subpath of a lift of~$p_i$, is otherwise disjoint from that lift, and projects to a closed curve~$c$ such that $\rho(\sigma(c,L))\in A$.
  
  Let $\tilde p_i$ be a lift of~$p_i$ and let $(v_0,\nu_0), (v_1,\nu_1), \ldots, (v_m,\nu_m)$ be the sequence of vertices on~$\tilde p_i$.  We consider the combinatorial surface~$\tilde\surf'$ obtained by cutting $\tilde \surf$ along the path~$\tilde p_i$, thus forming a boundary, and then attaching a disk~$b$ to the resulting boundary component; each interior vertex~$(v_t,\nu_t)$ of~$\tilde p_i$, $1\le t\le m-1$, now corresponds to two vertices, $(v_t,\nu_t)^+$ and~$(v_t,\nu_t)^-$; each duplicated edge has the same weight as its original. The vertices $(v_t,\nu_t)^\pm$ ($0\le t\le m$, where for convenience $(v_0,\nu_0)^\pm$ and $(v_m,\nu_m)^\pm$ denote $(v_0,\nu_0)$ and~$(v_m,\nu_m)$, respectively) all lie on the boundary of the face~$b$ of~$\tilde\surf'$.  By Lemma~\ref{l:endpoint}, it suffices to compute a shortest path, in this combinatorial surface, among all paths from some vertex $(v_t,\nu_t)^\pm$ to some corresponding vertex in the set $\Setbar{(v_t,\nu_t+a)}{a\in A}$.  Lemma~\ref{L:multiple-shortest-paths} allows us to do this: $(\tilde\surf,\tilde G)$ is a combinatorial surface of size $O(2^kn)$ and genus $O(2^kg)$ (by Lemma~\ref{L:rhocovers-algo}), and the same holds for the combinatorial surface in which we perform the computation; moreover, there are $O(2^kn)$ pairs of vertices between which we need to compute the distance.  Thus, the preprocessing step takes $O(g^22^{3k}kn\log n)$ time, and the distance computations take $O(2^kkn\log n)$.  Once we have computed a shortest distance between these pairs of points, we can compute an actual shortest path using Dijkstra's algorithm, without overhead.

  Applying this for each $i=1,\ldots,2g$, and returning the projection of the overall shortest path, we obtain the result.
\end{proof}

%%%%%%%%%%%%%%%%%%%%%%%%%%%%%%%%%%%%%%%%%%%%%%%%%%%%%%%%%%%%%%%%%%%%%
\section{Proof of Theorem~\ref{T:common-tool}}\label{S:proof-common-tool}

We will need the following intuitive lemma.
\begin{restatable}{lemma}{Lhomopreserve}\label{L:homopreserve}
  On a cross-metric surface~$(\surf,G^*)$, orientable or not, let $L$ be a system of loops and $c_1$ and~$c_2$ be two closed curves that cross each edge of~$G^*$ with the same parity.  Then $\sigma(c_1,L)=\sigma(c_2,L)$.
\end{restatable}
\begin{proof}
   The lemma is obvious with homology tools:  The assumption implies that $c_1$ and~$c_2$ are homologous, so the conclusion holds.  Below, we present an alternate, somewhat more self-contained, proof.

   Let $G$ be the dual graph of~$G^*$.  For $i=1,2$, we push $c_i$ homotopically to the closed walk~$c'_i$ in~$G$ defined as follows:  If $c_i$ crosses edges $e^*_1,\ldots,e^*_k$ of~$G^*$ in this order, then $c'_i$ traverses the corresponding edges $e_1,\ldots,e_k$ of~$G$ in the same order.  We have (1) $\sigma(c_i,L)=\sigma(c'_i,L)$; indeed, it is a folklore result that the parity of the number of crossings between two closed curves depends only on the homotopy of these closed curves (e.g., because every homotopy between two multicurves can be realized by Reidemeister moves~\cite{gs-mcmcr-97}).   Moreover, for each edge~$e$ of~$G$, the curves $c'_1$ and~$c'_2$ traverse $e$ with the same parity, which implies (2) $\sigma(c'_1,L)=\sigma(c'_2,L)$.  The result follows from (1) and~(2).
\end{proof}

\begin{proof}[Proof of Theorem~\ref{T:common-tool}]
  Let $(\surf,G^*)$ be our input cross-metric surface, and let $G$ be the dual graph of~$G^*$.  We start by computing, in $O(gn)$ time, a standard system of loops~$L'$ in general position with respect to~$G$ such that each loop crosses each edge of~$G$ at most 100 times; for this purpose, somewhat counterintuitively, we apply Proposition~\ref{P:std-sysloops} in the cross-metric surface~$(\surf,G)$.  For each edge $e$ of~$G$, we can now compute $\sigma(e,L')$, also in $O(gn)$ time. 
  
  Let $\rho':=\rho\circ\varphi$, where $\varphi$ is the map from Lemma~\ref{L:change-basis}.  For each edge~$e$ of~$G$, we compute $\rho'(\sigma(e,L'))$, which by Lemma~\ref{L:change-basis} equals $\rho(\sigma(e,L))$, for some fixed canonical system of loops~$L$.  This takes $O(kgn)$ time (there is no need to precompute the matrix of~$\rho'$ since its computation takes $O(kg^2)$ time). We then apply Proposition~\ref{P:subroutine} in the combinatorial surface~$(\surf,G)$ with the map~$\rho'$ and the standard system of loops~$L'$: In $O(g^38^kkn\log n)$ time, we compute a shortest closed walk~$c$ in~$G$ such that $\rho'(\sigma(c,L'))\in A$, or equivalently $\rho(\sigma(c,L))\in A$.

  The remaining part of the proof is to turn $c$ into a simple cycle.  First, we build a map~$\mu$ from the edges of~$G^*$ to~$\{0,1,2\}$ based on~$c$, as follows.  Let $e$ be an edge of~$G$.  If $c$ does not traverse~$e$, then we set $\mu(e)=0$.  If $c$ traverses~$e$ a positive, even number of times, then we set $\mu(e)=2$.  Otherwise, we set $\mu(e)=1$.  Then, we apply Lemma~\ref{L:single-cycle-odd-crossing} to $(\surf,G)$ and~$\mu$.  This computes in linear time a simple closed curve~$c'$ in the cross-metric surface~$(\surf,G^*)$ among those that cross each edge~$e^*$ of~$G^*$ exactly $\mu(e)$ times.  We return~$c'$.  Indeed, by definition of~$\mu$, it has multiplicity at most two and it is no longer than~$c$, which by Lemma~\ref{L:homopreserve} implies that it is a shortest closed curve such that $\rho(\sigma(c',L))\in A$.
\end{proof}

%%%%%%%%%%%%%%%%%%%%%%%%%%%%%%%%%%%%%%%%%%%%%%%%%%%%%%%%%%%%%%%%%%%%%%%%%%%%%%
\section{NP-Hardness of Computing a Shortest Orienting Closed Curve}\label{S:hardness}

In this section, we prove that it is NP-hard to decide, given a cross-metric surface and an integer~$k$, whether a shortest orienting closed curve has length at most~$k$ (Theorem~\ref{T:Nphardness}).  The proof is a variation on the proof by Chambers, Colin de Verdière, Erickson, Lazarus, and Whittlesey~\cite{ccelw-scsh-08} that computing the shortest splitting closed curve is NP-hard.  For this proof, it is easier to reason in the realm of combinatorial surfaces.

\begin{figure}
  \includegraphics[width=0.9\textwidth]{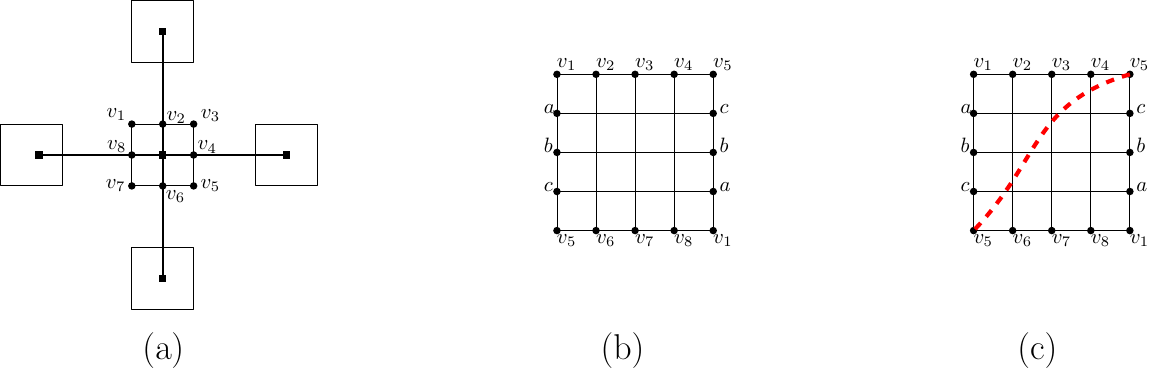}
  \caption{The reduction for the proof of Theorem~\ref{T:Nphardness}.  From left to right: (a) A part of the original grid graph~$G$, where each vertex is overlaid with cycles of length eight, each denoted $v_1,\ldots,v_8$; (b) the attached Möbius band (vertices marked $a$, $b$, $c$, $v_1$, and~$v_5$ are identified in pairs; the boundary of the Möbius band is thus cycle $v_1,\ldots,v_8$). (c) Part of the system of loops inside a Möbius band.}
\centering
\label{f:hamiltonian_1}
\end{figure}

\TNphardness*
\begin{proof}
  A \emph{grid graph} of size~$n$ is a graph induced by a set of $n$~points in the two-dimensional integer grid.  
  We reduce the problem of deciding whether there exists a Hamiltonian cycle in grid graphs, which is NP-hard~\cite{ips-hpgg-82}, to that of computing a shortest orienting closed curve.

  Let $G$ be a grid graph of size $n$ embedded in the sphere and without loss of generality assume that $n\geq 3$. First, we overlay $G$ with $n$ small cycles of length~8
  each centered on a vertex of~$G$; see Figure~\ref{f:hamiltonian_1}(a).  We then remove the interior of these cycles, obtaining a surface of genus zero with $n$~boundary components.
  For each of these boundary components we associate a distinct Möbius band coming from a $4\times4$-grid as in Figure~\ref{f:hamiltonian_1}(b). 
  We identify the boundary of each Möbius band with the associated boundary component\footnote{To illustrate this identification, let $v_1,\dots,v_8$ denote the vertices of a single boundary component of the surface as in Figure~\ref{f:hamiltonian_1}(a) and of the boundary of the associated Möbius band as in Figure~\ref{f:hamiltonian_1}(b). We proceed by identifying the edge $v_iv_{i+1}$ of the boundary component with the edge $v_iv_{i+1}$ of the associated Möbius band.}.    
  Let $G'$ be the resulting graph, which is embedded in a non-orientable surface~$\surf$ of genus~$n$.  We assign weights to each edge of~$G'$:  Each edge in a Möbius band (including those edges in the cycles of length~8) has weight~$\varepsilon<1/(12n)$, and every other edge has weight one.  Thus, $(\surf,G')$ is a combinatorial surface with size $O(n)$.   We claim that $G$ has a Hamiltonian cycle if and only if there exists an orienting closed curve of length at most $n+1/2$ in the combinatorial surface~$(\surf,G')$. We prove this below, which in turn will conclude the proof.

  As an auxiliary tool, we introduce a canonical system of loops~$L$ based at an arbitrary basepoint of~$\surf$, each loop associated to a given vertex of~$G$.
  More precisely, for a given vertex of $G$ let $v_1,\dots,v_8$ be the vertices of its corresponding $8$-cycle, as in Figure~\ref{f:hamiltonian_1}(a), then its associated loop is a small perturbation of the concatenation of a loop that goes from the basepoint to~$v_5$, goes ``through the crosscap'' from one copy of~$v_5$ to the other copy, and finally goes back from~$v_5$ to the basepoint via the same path, see Figure~\ref{f:hamiltonian_1}(c).
  
\begin{figure}
  \includegraphics[width=0.9\textwidth]{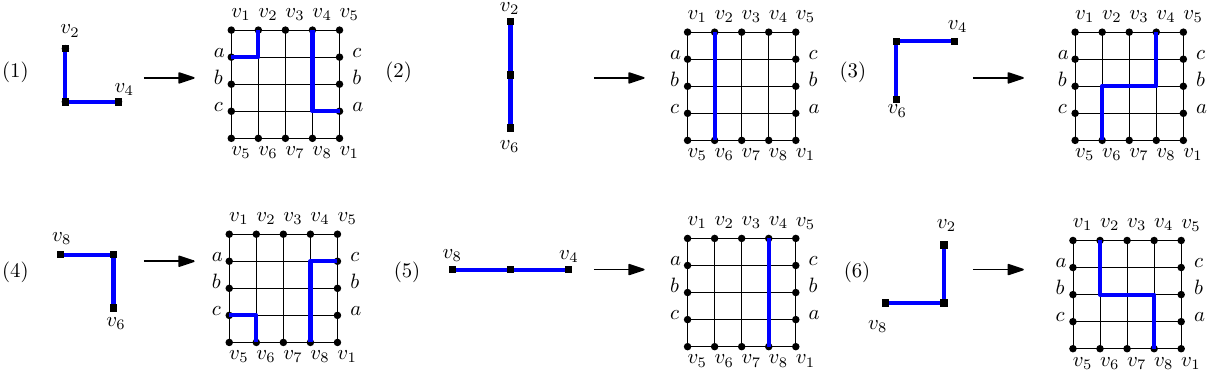}
  \caption{Local modifications at every vertex of a Hamiltonian cycle of~$G$ to turn it into a simple orienting cycle in the combinatorial surface $(\surf,G)$.}
\centering
\label{f:hamiltonian_2}
\end{figure}

  If $G$ has a Hamiltonian cycle, thus of length~$n$ in~$G$, we can easily modify that cycle (Figure~\ref{f:hamiltonian_2}) in such a way that it becomes a simple cycle of length at most~$n+6n\varepsilon<n+1/2$ in~$G'$ crossing each loop in~$L$ exactly once, and thus orienting (Lemma~\ref{L:top-charact}).

  Conversely, any orienting cycle~$c$ in $(\surf,G')$ of length at most $n+1/2$ must cross each loop in~$L$ an odd number of times (Lemma~\ref{L:top-charact}), and thus
  pass through each of the $n$ Möbius bands.
  Moreover, it uses at most $n$ edges of weight $1$ in~$G'$.  Thus, it corresponds, in~$G$, to a Hamiltonian cycle. Indeed, since $n\geq 3$, any closed walk without vertex repetition in~$G$ is a Hamiltonian cycle in~$G$. 
\end{proof}

%%%%%%%%%%%%%%%%%%%%%%%%%%%%%%%%%%%%%%%%%%%%%%%%%%%%%%%%%%%%%
\section{Conclusion}

We conclude with several remarks:
\begin{itemize}
\item Our main tool, Theorem~\ref{T:common-tool}, also holds for orientable surfaces.  In that case, the proof is simpler, bypassing the detour with standard systems of loops:  We compute a canonical system of loops using Lemma~\ref{L:lpvv}, compute the signature of every edge with respect to that system of loops, apply Proposition~\ref{P:subroutine}, and conclude as in the non-orientable case.  
\item There is a variant of our algorithm for non-orientable surfaces in which we can completely bypass the computation of a standard system of loops.  Let us sketch the algorithm here.  First, we compute an arbitrary system of loops~$L'$, in $O(gn)$ time.  Then, we compute the ``change of coordinates'' matrices between~$L'$ and some canonical system of loops~$L$, as in the statement of Lemma~\ref{L:change-basis}; because $L'$ is not fixed, we cannot make such matrices explicit, but since our proof is algorithmic (and indeed follows the proof of the classification theorem by Brahana~\cite{b-sc2dm-21,stillwell-topology-93}), we can actually do this using the same principles, in time polynomial in~$g$ (without dependence on~$n$).  We emphasize that, after each cut-and-paste step, it is enough to remember only the signature vector $\sigma(\ell,L')$ for each loop~$\ell$ in the current system of loops.  More precisely, $L'$ is a polygonal schema of size $O(g)$; one can, in $O(g^2)$ time, obtain a sequence of $O(g)$ cut-and-paste operations to transform $L'$ to~$L$.  Each such cut-and-paste operation replaces a loop, say~$\ell_1$, with a new one, say~$\ell_2$; computing $\sigma(\ell_2,L')$ takes $O(g^2)$ time, because $\ell_2$ is homotopic to the concatenation of $O(g)$ loops in the current system of loops, for which we know the signatures with respect to~$L'$.  Thus, in $O(g^3)$ total time, one can compute $\sigma(\ell,L')$ for all the loops $\ell\in L$, and then using matrix inversion on $g\times g$-matrices one can compute $\sigma(\ell',L)$ for each loop $\ell'\in L'$, which are the values we need in the proof of Theorem~\ref{T:common-tool}.

The version presented here, while perhaps a bit longer, has some benefits compared to the route sketched above: (1) Our Propositions~\ref{P:matousek} and~\ref{P:std-sysloops} are of independent interest and are likely to be reused in future work; (2) having an explicit relation on the homology bases of the standard and canonical systems of loops is somewhat nicer conceptually.
\item Our techniques allow to compute a shortest overall one-sided closed curve in $O(g^3n\log n)$ time (without controlling whether it is orienting or not); indeed, just apply Theorem~\ref{T:common-tool} with $\rho(c_1,\ldots,c_g)=\sum_ic_i$ and $A=\Set{1}$.
\item Our algorithms run in $O(n\log n)$ for fixed genus.  It might be possible, in the case of two-sided curves, to obtain an algorithm with running time $O(n\log\log n)$, using the techniques by Chambers, Erickson, Fox, and Nayyeri~\cite[Section~4.3]{cefn-mcsg-23}, though with a hidden dependence on the genus that is at least exponential.
\item Finally, we have only considered non-separating closed curves, and leave open the complexity of computing a shortest separating closed curve with a given topological type (specifying the topology of the two surfaces with boundary resulting from cutting along it).   Even on orientable surfaces, the following problem is fixed-parameter tractable in the genus~\cite[Theorem~6.1]{ccelw-scsh-08} but apparently neither known to be NP-hard nor polynomial-time solvable: Given an orientable surface~$\surf$ of (orientable) genus~$g$, compute a shortest simple closed curve that separates~$\surf$ into a surface of (orientable) genus one and a surface of (orientable) genus $g-1$.
\end{itemize}

\subsection*{Acknowledgments} 
We would like to thank Arnaud de Mesmay for stimulating discussions, and the anonymous reviewers for their useful comments.

%%%%%%%%%%%%%%%%%%%%%%%%%%%%%%%%%%%%%%%%%%%%%%%%%%%%%%%%%%%%%


\begin{thebibliography}{10}

\bibitem{ah-irgg-77}
Michael~O. Albertson and Joan~P. Hutchinson.
\newblock The independence ratio and genus of a graph.
\newblock {\em Trans. Amer. Math. Soc.}, 226:161--173, 1977.
\newblock \href {https://doi.org/10.2307/1997946} {\path{doi:10.2307/1997946}}.

\bibitem{a-bt-83}
Mark~A. Armstrong.
\newblock {\em Basic topology}.
\newblock Undergraduate Texts in Mathematics. Springer-Verlag, New York-Berlin,
  1983.
\newblock \href {https://doi.org/10.1007/978-1-4757-1793-8}
  {\path{doi:10.1007/978-1-4757-1793-8}}.

\bibitem{bcfn-mchbs-17}
Glencora Borradaile, Erin~W. Chambers, Kyle Fox, and Amir Nayyeri.
\newblock Minimum cycle and homology bases of surface-embedded graphs.
\newblock {\em J. Comput. Geom.}, 8(2):58--79, 2017.
\newblock \href {https://doi.org/10.20382/jocg.v8i2a4}
  {\path{doi:10.20382/jocg.v8i2a4}}.

\bibitem{b-sc2dm-21}
Henry~R. Brahana.
\newblock Systems of circuits on two-dimensional manifolds.
\newblock {\em Ann. of Math. (2)}, 23(2):144--168, 1921.
\newblock \href {https://doi.org/10.2307/1968030} {\path{doi:10.2307/1968030}}.

\bibitem{cce-msspe-13}
Sergio Cabello, Erin~W. Chambers, and Jeff Erickson.
\newblock Multiple-source shortest paths in embedded graphs.
\newblock {\em SIAM J. Comput.}, 42(4):1542--1571, 2013.
\newblock \href {https://doi.org/10.1137/120864271}
  {\path{doi:10.1137/120864271}}.

\bibitem{cdem-fotc-10}
Sergio Cabello, Matt Devos, Jeff Erickson, and Bojan Mohar.
\newblock Finding one tight cycle.
\newblock {\em ACM Trans. Algorithms}, 6(4):Article~61, 2010.
\newblock \href {https://doi.org/10.1145/1824777.1824781}
  {\path{doi:10.1145/1824777.1824781}}.

\bibitem{ccelw-scsh-08}
Erin~W. Chambers, {\'E}ric Colin~de Verdi{\`e}re, Jeff Erickson, Francis
  Lazarus, and Kim Whittlesey.
\newblock Splitting (complicated) surfaces is hard.
\newblock {\em Comput. Geom.}, 41(1-2):94--110, 2008.
\newblock \href {https://doi.org/10.1016/j.comgeo.2007.10.010}
  {\path{doi:10.1016/j.comgeo.2007.10.010}}.

\bibitem{cefn-mcsg-23}
Erin~W. Chambers, Jeff Erickson, Kyle Fox, and Amir Nayyeri.
\newblock Minimum cuts in surface graphs.
\newblock {\em SIAM J. Comput.}, 52(1):156--195, 2023.
\newblock \href {https://doi.org/10.1137/19M1291820}
  {\path{doi:10.1137/19M1291820}}.

\bibitem{cex-dwsp-15}
Hsien-Chih Chang, Jeff Erickson, and Chao Xu.
\newblock Detecting weakly simple polygons.
\newblock In {\em Proceedings of the 26th Annual ACM-SIAM Symposium on Discrete
  Algorithms (SODA)}, pages 1655--1670. Society for Industrial and Applied
  Mathematics, 2015.
\newblock \href {https://doi.org/10.1137/1.9781611973730.110}
  {\path{doi:10.1137/1.9781611973730.110}}.

\bibitem{c-scgsp-10}
{\'E}ric Colin~de Verdi{\`e}re.
\newblock Shortest cut graph of a surface with prescribed vertex set.
\newblock In {\em Proceedings of the 18th European Symposium on Algorithms
  (ESA), part 2}, volume 6347 of {\em Lecture Notes in Computer Science}, pages
  100--111. Springer-Verlag, Berlin, 2010.
\newblock \href {https://doi.org/10.1007/978-3-642-15781-3\_9}
  {\path{doi:10.1007/978-3-642-15781-3\_9}}.

\bibitem{c-ctgs-18}
{\'E}ric Colin~de Verdi{\`e}re.
\newblock Computational topology of graphs on surfaces.
\newblock In Jacob~E. Goodman, Joseph O'Rourke, and Csaba Toth, editors, {\em
  Handbook of Discrete and Computational Geometry}, chapter~23, pages 605--636.
  CRC Press LLC, third edition, 2018.
\newblock \href {https://doi.org/10.1201/9781315119601}
  {\path{doi:10.1201/9781315119601}}.

\bibitem{ce-tnpcs-10}
{\'E}ric Colin~de Verdi{\`e}re and Jeff Erickson.
\newblock Tightening nonsimple paths and cycles on surfaces.
\newblock {\em SIAM J. Comput.}, 39(8):3784--3813, 2010.
\newblock \href {https://doi.org/10.1137/090761653}
  {\path{doi:10.1137/090761653}}.

\bibitem{gs-mcmcr-97}
Maurits de~Graaf and Alexander Schrijver.
\newblock Making curves minimally crossing by {R}eidemeister moves.
\newblock {\em J. Combin. Theory Ser. B}, 70(1):134--156, 1997.
\newblock \href {https://doi.org/10.1006/jctb.1997.1754}
  {\path{doi:10.1006/jctb.1997.1754}}.

\bibitem{e-dgteg-03}
David Eppstein.
\newblock Dynamic generators of topologically embedded graphs.
\newblock In {\em Proceedings of the 14th Annual ACM-SIAM Symposium on Discrete
  Algorithms (SODA)}, pages 599--608. Society for Industrial and Applied
  Mathematics, 2003.

\bibitem{e-sntcd-11}
Jeff Erickson.
\newblock Shortest non-trivial cycles in directed surface graphs.
\newblock In {\em Proceedings of the 27th Annual Symposium on Computational
  Geometry (SOCG)}, pages 236--243. Association for Computing Machinery, 2011.
\newblock \href {https://doi.org/10.1145/1998196.1998231}
  {\path{doi:10.1145/1998196.1998231}}.

\bibitem{efl-hmcpf-18}
Jeff Erickson, Kyle Fox, and Luvsandondov Lkhamsuren.
\newblock Holiest minimum-cost paths and flows in surface graphs.
\newblock In {\em Proceedings of the 50th Annual ACM SIGACT Symposium on Theory
  of Computing (STOC)}, pages 1319--1332. Association for Computing Machinery,
  2018.
\newblock \href {https://doi.org/10.1145/3188745.3188904}
  {\path{doi:10.1145/3188745.3188904}}.

\bibitem{eh-ocsd-04}
Jeff Erickson and Sariel Har-Peled.
\newblock Optimally cutting a surface into a disk.
\newblock {\em Discrete Comput.\ Geom.}, 31(1):37--59, 2004.
\newblock \href {https://doi.org/10.1007/s00454-003-2948-z}
  {\path{doi:10.1007/s00454-003-2948-z}}.

\bibitem{en-mcsnc-11}
Jeff Erickson and Amir Nayyeri.
\newblock Minimum cuts and shortest non-separating cycles via homology covers.
\newblock In {\em Proceedings of the 22nd Annual ACM-SIAM Symposium on Discrete
  Algorithms (SODA)}, pages 1166--1176. Society for Industrial and Applied
  Mathematics, 2011.

\bibitem{ew-gohhg-05}
Jeff Erickson and Kim Whittlesey.
\newblock Greedy optimal homotopy and homology generators.
\newblock In {\em Proceedings of the 16th Annual ACM-SIAM Symposium on Discrete
  Algorithms (SODA)}, pages 1038--1046. Association for Computing Machinery,
  2005.

\bibitem{ew-csec-10}
Jeff Erickson and Pratik Worah.
\newblock Computing the shortest essential cycle.
\newblock {\em Discrete Comput. Geom.}, 44(4):912--930, 2010.
\newblock \href {https://doi.org/10.1007/s00454-010-9241-8}
  {\path{doi:10.1007/s00454-010-9241-8}}.

\bibitem{fm-pmcg-11}
Benson Farb and Dan Margalit.
\newblock {\em A primer on mapping class groups}, volume~49 of {\em Princeton
  Mathematical Series}.
\newblock Princeton University Press, Princeton, NJ, 2012.

\bibitem{f-sntcd-13}
Kyle Fox.
\newblock Shortest non-trivial cycles in directed and undirected surface
  graphs.
\newblock In {\em Proceedings of the 24th Annual ACM-SIAM Symposium on Discrete
  Algorithms (SODA)}, pages 352--364. Society for Industrial and Applied
  Mathematics, 2013.

\bibitem{fuladi2023short}
Niloufar Fuladi, Alfredo Hubard, and Arnaud de~Mesmay.
\newblock Short topological decompositions of non-orientable surfaces.
\newblock {\em Discrete Comput. Geom.}, 72(2):783--830, 2024.
\newblock \href {https://doi.org/10.1007/s00454-023-00580-3}
  {\path{doi:10.1007/s00454-023-00580-3}}.

\bibitem{gavoille2023minor}
Cyril Gavoille and Claire Hilaire.
\newblock Minor-universal graph for graphs on surfaces.
\newblock {\em arXiv preprint arXiv:2305.06673}, 2023.
\newblock \href {https://doi.org/10.48550/arXiv.2305.06673}
  {\path{doi:10.48550/arXiv.2305.06673}}.

\bibitem{g-frm-83}
Mikhael Gromov.
\newblock Filling {R}iemannian manifolds.
\newblock {\em J. Differential Geom.}, 18(1):1--147, 1983.
\newblock \href {https://doi.org/10.4310/jdg/1214509283}
  {\path{doi:10.4310/jdg/1214509283}}.

\bibitem{gross-tucker87-graph}
Jonathan~L. Gross and Thomas~W. Tucker.
\newblock {\em Topological graph theory}.
\newblock Wiley-Interscience Series in Discrete Mathematics and Optimization.
  John Wiley \& Sons, Inc., New York, 1987.

\bibitem{ips-hpgg-82}
Alon Itai, Christos~H. Papadimitriou, and Jayme~Luiz Szwarcfiter.
\newblock Hamilton paths in grid graphs.
\newblock {\em SIAM J. Comput.}, 11(4):676--686, 1982.
\newblock \href {https://doi.org/10.1137/0211056} {\path{doi:10.1137/0211056}}.

\bibitem{lpvv-ccpso-01}
Francis Lazarus, Michel Pocchiola, Gert Vegter, and Anne Verroust.
\newblock Computing a canonical polygonal schema of an orientable triangulated
  surface.
\newblock In {\em Proceedings of the 17th Annual Symposium on Computational
  Geometry (SOCG)}, pages 80--89. Association for Computing Machinery, 2001.
\newblock \href {https://doi.org/10.1145/378583.378630}
  {\path{doi:10.1145/378583.378630}}.

\bibitem{l-gem-82}
S\'{o}stenes Lins.
\newblock Graph-encoded maps.
\newblock {\em J. Combin. Theory Ser. B}, 32(2):171--181, 1982.
\newblock \href {https://doi.org/10.1016/0095-8956(82)90033-8}
  {\path{doi:10.1016/0095-8956(82)90033-8}}.

\bibitem{matouvsek2016untangling}
Ji\v{r}\'{\i} Matou\v{s}ek, Eric Sedgwick, Martin Tancer, and Uli Wagner.
\newblock Untangling two systems of noncrossing curves.
\newblock {\em Israel J. Math.}, 212(1):37--79, 2016.
\newblock \href {https://doi.org/10.1007/s11856-016-1294-9}
  {\path{doi:10.1007/s11856-016-1294-9}}.

\bibitem{mohar2009genus}
Bojan Mohar.
\newblock The genus crossing number.
\newblock {\em Ars Math. Contemp.}, 2(2):157--162, 2009.
\newblock \href {https://doi.org/10.26493/1855-3974.21.157}
  {\path{doi:10.26493/1855-3974.21.157}}.

\bibitem{rs-gm16e-03}
Neil Robertson and P.~D. Seymour.
\newblock Graph minors. {XVI}. {E}xcluding a non-planar graph.
\newblock {\em J. Combin. Theory Ser. B}, 89(1):43--76, 2003.
\newblock \href {https://doi.org/10.1016/S0095-8956(03)00042-X}
  {\path{doi:10.1016/S0095-8956(03)00042-X}}.

\bibitem{JGAA-580}
Marcus Schaefer and Daniel \v{S}tefankovi\v{c}.
\newblock The degenerate crossing number and higher-genus embeddings.
\newblock {\em J. Graph Algorithms Appl.}, 26(1):35--58, 2022.
\newblock \href {https://doi.org/10.7155/jgaa.00579}
  {\path{doi:10.7155/jgaa.00579}}.

\bibitem{stillwell-topology-93}
John Stillwell.
\newblock {\em Classical topology and combinatorial group theory}, volume~72 of
  {\em Graduate Texts in Mathematics}.
\newblock Springer-Verlag, New York, second edition, 1993.
\newblock \href {https://doi.org/10.1007/978-1-4612-4372-4}
  {\path{doi:10.1007/978-1-4612-4372-4}}.

\bibitem{t-egsnc-90}
Carsten Thomassen.
\newblock Embeddings of graphs with no short noncontractible cycles.
\newblock {\em J. Combin. Theory Ser. B}, 48(2):155--177, 1990.
\newblock \href {https://doi.org/10.1016/0095-8956(90)90115-G}
  {\path{doi:10.1016/0095-8956(90)90115-G}}.

\end{thebibliography}
\end{document}